%% file: main.tex
\documentclass[conference]{IEEEtran}
\IEEEoverridecommandlockouts

\usepackage{cite}
\usepackage{amsmath,amssymb,amsfonts}
\usepackage{algorithmic}
\usepackage{graphicx}
\usepackage{textcomp}
\usepackage{xcolor}
\usepackage{xspace}
\usepackage[all]{xy} \SelectTips{cm}{}  %
\usepackage{mathtools}
\usepackage{tikz-cd, tikz}
\usepackage{stmaryrd}
\usepackage{mdframed}
\def\BibTeX{{\rm B\kern-.05em{\sc i\kern-.025em b}\kern-.08em
    T\kern-.1667em\lower.7ex\hbox{E}\kern-.125emX}}
\usepackage[bookmarks=true, hidelinks, pdfencoding=auto]{hyperref}

\newif\ifdraft\draftfalse
\newif\ifarxiv\arxivtrue

\ifdraft
\usepackage[colorinlistoftodos]{todonotes}

\else
\usepackage[disable]{todonotes}
\fi

\usepackage{amsthm}
\theoremstyle{definition}
\newtheorem{definition}{Definition}[section]
\theoremstyle{plain}
\newtheorem{lemma}[definition]{Lemma}
\newtheorem{proposition}[definition]{Proposition}
\newtheorem{theorem}[definition]{Theorem}
\newtheorem{corollary}[definition]{Corollary}

\theoremstyle{definition}
\newtheorem{remark}[definition]{Remark}

\newtheorem{example}[definition]{Example}

\newtheorem{notation}[definition]{Notation}

\theoremstyle{remark}

\newcommand{\rew}{\mathrm{rew}}

\newcommand{\Coalg}[1]{\mathrm{Coalg}(#1)}
\newcommand{\CoalgRst}[2]{\mathrm{Coalg}(#1)|_{#2}}
\newcommand{\Alg}[1]{\mathrm{Alg}(#1)}
\newcommand{\AlgRst}[2]{\mathrm{Alg}(#1)|_{#2}}

\newcommand{\CAT}{\mathbf{CAT}}
\newcommand{\op}{\mathrm{op}}

\newcommand{\adj}[4]{(#1, #2; #3, #4)} %
\newcommand{\mcfunc}[1]{\mathcal{D}(#1 + \{\checkmark\})}

\newcommand{\fixunit}[1]{#1_\mathrm{fix}}
\newcommand{\fixcounit}[1]{#1_{\mathrm{fix}}}
\newcommand{\AlgFn}[2]{\mathrm{Alg}(#1)_{#2}}
\newcommand{\CoalgFn}[2]{\mathrm{Coalg}(#1)_{#2}}

\newcommand{\interval}[1]{\mathbb{I}_{#1}}
\newcommand{\Prb}[1]{\mathrm{P}_{#1}}
\newcommand{\Erew}[1]{\mathrm{Erew}_{#1}}

\newcommand{\gra}{global reachability condition\xspace}
\newcommand{\fibop}[1]{#1^\mathrm{fibop}}
\newcommand{\laxdom}[1]{d_{#1}}

\newcommand{\rloop}[2][-]{\save \POS!R(.7) \ar@(ru,rd)^#1{#2} \restore}
\newcommand{\lloop}[2][-]{\save \POS!L(.7) \ar@(lu,ld)_#1{#2} \restore}
\newcommand{\uloop}[2][-]{\save \POS!U(.7) \ar@(lu,ru)^(.8){#2} \restore}

\newcommand{\SupLat}{\mathbf{SupLat}}
\newcommand{\Poset}{\mathbf{Poset}}
\newcommand{\Set}{\mathbf{Set}}

\newcommand{\disc}{\mathrm{disc}}
\newcommand{\id}{\mathrm{id}}
\newcommand{\Pfin}{{\mathcal{P}_f}}

\newcommand{\pullbackmark}[2]{\save ;p+<.8pc,0pc>:(0,-1)::%
(#1) *{\phantom{Z}} %
;p+(#2)-(0,0) **@{-}%
;p-(#1)+(0,0) *{\phantom{Z}} **@{-} \restore}

\begin{document}

\title{Initial Algebra Correspondence \\under Reachability Conditions}

\author{\IEEEauthorblockN{Mayuko Kori}
\IEEEauthorblockA{
\textit{Research Institute for Mathematical Sciences} \\
\textit{Kyoto University}\\
Kyoto, Japan \\
mkori@kurims.kyoto-u.ac.jp}
\and
\IEEEauthorblockN{Kazuki Watanabe}
\IEEEauthorblockA{
\textit{National Institute of Informatics, The Graduate}\\
\textit{University for Advanced Studies, SOKENDAI}\\
Tokyo, Japan \\
kazukiwatanabe@nii.ac.jp}
\and
\IEEEauthorblockN{Jurriaan Rot}
\IEEEauthorblockA{\textit{Radboud University} \\
Nijmegen, The Netherlands \\
jrot@cs.ru.nl}
}

\maketitle

\begin{abstract}
  Suitable reachability conditions can make two different fixed point
  semantics of a transition system coincide.
  For instance, the total and partial expected reward semantics
  on Markov chains (MCs) coincide whenever the MC at hand
  is almost surely reachable.
  In this paper, we present a unifying framework for such reachability conditions
  that ensures the correspondence of two different semantics.
  Our categorical framework naturally induces an abstract reachability condition via a suitable adjunction,
  which allows us to prove coincidences of fixed points, and more generally of initial algebras.
  We demonstrate the generality of our approach by instantiating several examples,
  including the almost sure reachability condition for MCs, and the unambiguity condition of automata.
  We further study a canonical construction of our instance for Markov decision processes by pointwise Kan extensions.
\end{abstract}

\begin{IEEEkeywords}
initial algebra, reachability, adjunction.
\end{IEEEkeywords}

\section{Introduction}

Fixed points of predicate transformers occur everywhere in the semantics of transition systems.
Predicate transformers describe a one-step update of the semantics, and can be seen as specifications of the \emph{local behaviour} of transition systems.
The Kleene fixed-point theorem says that the least fixed point of a predicate transformer on a complete lattice is the supremum of the $\omega$-chain that is constructed by iterating the predicate transformer.
The converged fixed points of such local updates describe the \emph{global behaviours} of systems.

In practice, the design of suitable semantics is crucial for formal verification tasks.
Indeed, there are many options for such semantics even if we restrict to very specific objectives, say expected rewards on probabilistic systems: e.g. total expected rewards~\cite{Puterman94, DBLP:books/daglib/0020348}, partial expected rewards~\cite{DBLP:books/daglib/0020348}, conditional expected rewards~\cite{DBLP:books/daglib/0020348}, and liberal expected rewards~\cite{DBLP:conf/qest/GretzKM12,OlmedoGJKKM18}.
In particular, the partial
expected reward captures the total reward along paths that reach a target state, whereas the total expected
reward captures the reward along all paths.
Such differences give rise to different applications, for instance, total expected rewards for expected runtimes~\cite{KaminskiKMO18}, partial expected rewards for distributed algorithms~\cite{ChenFKPS13}, and conditional and liberal expected rewards for conditioning~\cite{OlmedoGJKKM18}.

Category theory gives a general framework for semantics of systems beyond the above lattice-theoretic framework.
For instance,
least fixed points become initial algebras of associated predicate transformers modelled as functors on suitable posetal categories.
Semantics of systems is then modelled by these initial algebras.
Dually, the theory of coalgebras provides a methodology of how to model transition systems and their associated predicate transformers in a uniform way~\cite{DBLP:books/cu/J2016}.

The existing theory of coalgebras, however, is still not sufficiently developed to capture certain common phenomena that occur in formal verification.
The phenomenon we are interested in this paper is the correspondence of two different semantics under certain \emph{reachability conditions}.
More precisely, reachability conditions on global behaviour of systems may simplify semantics, and consequently, the difference between the two semantics may collapse.
In particular, under reachability conditions to target states (or equivalently termination assumptions), total and partial semantics do coincide in many cases, as partial semantics only cares about terminating executions.
For instance, the condition of \emph{almost sure reachability} on Markov chains (MCs) makes total expected rewards and partial expected rewards coincide (e.g.~\cite{DBLP:books/daglib/0020348}).

In this paper, we present a categorical framework for such correspondences of two semantics of systems under reachability conditions.
Specifically, we model systems as coalgebras, and semantics of systems as initial algebras of associated predicate transformers.
We focus on the initial algebra correspondences given by an adjunction between semantic domains, on which predicate transformers are defined.
Such adjunctions---which often represent translations between different semantic domains---are a standard setting in abstract interpretation, where they appear as Galois connections~\cite{cousot21}.

This setting allows us to present our main contributions, on initial algebra correspondences under
\emph{\gra{}s}, which are abstract reachability conditions with adjunctions.
Our approach is based on a result due to Hermida and Jacobs~\cite[Thm.~2.14]{DBLP:journals/iandc/HermidaJ98} on lifting adjunctions to categories of (co)algebras.
The key technical idea here is
to restrict these adjunctions to
specific subcategories of algebras.
These restricted adjunctions form equivalences, which yields our main result for initial algebra correspondences.
We instantiate our framework to the above-mentioned examples on expected rewards in MCs and Markov decision processes (MDPs),
and also to more elementary examples such as resource-bounded reachability in labeled transition systems.

We then extend the setting to initial algebra correspondences under a chain correspondence,
instead of relying solely on an equivalence of subcategories.
 This chain-based approach allows us to capture a wider variety of scenarios, and shows, for instance, how unambiguity of nondeterministic finite automata (NFAs) arises in our framework.

In the final part of the paper, we focus on an orthogonal issue: the construction of predicate transformers themselves.
In particular, we show how the predicate transformer, also called \emph{Bellman operator}, for partial expected rewards in MDPs (Example~\ref{eg:achievable_mdp}) arises from the one for MCs,
through suitable pointwise Kan extensions.

In summary, the contributions of this paper are as follows:
\begin{itemize}
	\item A general framework for proving initial algebra correspondences under \gra{}s;
	\item An extension of our framework employing initial chain correspondences, which can cover a different instance, namely an unambiguity condition in automata;
	\item A generic construction of predicate transformers from MCs to MDPs via pointwise Kan extensions;
\end{itemize}
Before presenting the general theory, in the next section we give an overview by means of an illustrative example: the coincidence
of total and partial expected rewards in MCs under almost sure reachability.

\paragraph*{Related work}
In previous work, particularly in abstract interpretation,
researchers have often studied one-step, and consequently \emph{local} conditions on predicate transformers
to establish the coincidence of two semantics connected by an adjunction.
A well-known example is the (backward) completeness~\cite{DBLP:conf/popl/CousotC79, DBLP:journals/jacm/GiacobazziRS00} (as well as local completeness~\cite{BruniGGR21,BruniGGR22}) of an abstract domain in abstract interpretation.
It
guarantees that an adjoint maps one initial algebra in a \emph{concrete domain} to its counterpart in an \emph{abstract domain}.
Our correspondence further requires the converse, and we ensure the correspondence by \gra{}s.
We provide instances of initial algebra correspondences in non-complete settings, as detailed in Remark~\ref{rm:complete}.

As mentioned above, our work builds on a result of
lifting adjunctions to categories of (co)algebras,
which immediately yields
preservation of initial algebras and final coalgebras due to Hermida and Jacobs~\cite{DBLP:journals/iandc/HermidaJ98}.
In the literature on coalgebras, the dual problem of final coalgebra preservation has been considered in general terms, e.g.,~\cite{DBLP:conf/calco/TurkenburgBKR23,DBLP:journals/logcom/SprungerKDH21}.
In the current paper, the focus instead is on situations where such correspondences do not hold in general but only under reachability conditions.

Reachability of coalgebras is studied in, e.g.,~\cite{DBLP:conf/fossacs/WissmannDKH19,Wissmann2019}. This generalizes reachability in an automata-theoretic sense (all states can be reached
from the initial state). Here, our interpretation of reachability is mostly on all executions terminating, or reaching specific target states.

\begin{notation}
  \label{notation}
  We write $\mathbf{2}$ for the lattice $\{\bot, \top\}$ of Booleans with the order satisfying $\bot  < \top$.
  We write $\interval{1}$ for the interval $[0, 1]$ with the usual order,
  and $\interval{\infty}$ for the set of non-negative extended reals, also equipped with the usual order.
  For sums and products, we adopt the convention $n + \infty = \infty + n = \infty$ and $0 \cdot \infty = \infty \cdot 0 = 0$ for each $n \in (\interval{\infty} \setminus \{0\})$ (see e.g.~\cite{Batz2025}).

  Both the composition of morphisms and
  the vertical composition of natural transformations are denoted by $\circ$.
  We write $\adj L R \eta \epsilon \colon {\mathbb{C}} \to {\mathbb{D}}$
  for the adjunction
  $L \dashv R\colon \mathbb{D} \to \mathbb{C}$
  with the unit $\eta\colon \id \Rightarrow RL$ and counit $\epsilon\colon LR \Rightarrow \id$.
  The adjoint transpose is written as $\overline{(-)}$.
  We write $[S, K]$ for the functor category.
  We often view a set $S$ as a discrete category
  and a complete lattice $(K, \sqsubseteq)$ as a category by regarding the order $\sqsubseteq$ as the existence of a unique morphism.
  For a set $S$ and a complete lattice $K$, we also write $[S, K]$ for the complete lattice of functions with the pointwise order.
  We write $\pi_1 $ and $\pi_2 $ for the first projection $\pi_1 \colon X\times Y\rightarrow X$ and the second projection $\pi_2 \colon X\times Y\rightarrow Y$, respectively.

\end{notation}

\section{Overview}
\label{sec:overview}
In this section, we outline the target problem, namely the correspondence of two definitions of semantics under reachability conditions, with a concrete example---Markov chains (MCs) with the almost sure reachability condition.
We then sketch our general framework for the problem, and illustrate our \gra{}s.

\subsection{MCs and Expected Rewards}
\label{subsec:MCsExpRew}
MCs are probabilistic transition systems that are commonly used in probabilistic verification~\cite{DBLP:books/daglib/0020348}.
We define an MC over a set $S$ of states and a target state $\checkmark$ as a pair of functions $(P, \rew)$ such that (i) $P\colon S\to \mathcal{D}(S + \{\checkmark\})$ gives the transition probability $P(s) \in\mathcal{D}(S + \{\checkmark\})$ for each state $s\in S$, where $\mathcal{D}(S + \{\checkmark\})$ is the set of distributions on $S + \{\checkmark\}$ whose supports are finite;
and (ii) $\rew\colon S\to\mathbb{N} $ assigns the reward $\rew(s)\in \mathbb{N}$ for each state $s\in S$.
We note that the target state $\checkmark$ is a dead end (or equivalently sink), without outgoing transitions.
MCs can be concisely written as coalgebras $c\coloneqq \langle P, \rew \rangle\colon S\to  \mathcal{D}(S + \{\checkmark\})\times \mathbb{N}$ of the endofunctor $\mathcal{D}(\_ + \{\checkmark\})$ on $\Set$.

Computing (or approximating) expected rewards on MCs is one of the most important problems for probabilistic verification.
It may, however, not be trivial what the appropriate definition of expected rewards is for a given model.
In fact, there are two commonly studied expected reward objectives: One is \emph{total}, and the other is \emph{partial}.

Perhaps the total expected reward objective (without discount factors) is more common, in particular for the expected runtimes of probabilistic programs~\cite{Puterman94,Batz2025}. This definition includes expectation of the rewards for all paths---importantly, this includes paths which never reach $\checkmark$.
The partial expected reward objective~\cite{DBLP:books/daglib/0020348,ChenFKPS13,Baier0KW17,PiribauerB19}, in contrast, only concerns paths that eventually reach $\checkmark$.
An application is the optimization of distributed algorithms that may not terminate (e.g.~\cite{ChenFKPS13}).
We refer to~\cite{DBLP:books/daglib/0020348} for a comprehensive overview.

\subsection{Total and Partial Expected Rewards Correspondence under the Almost Sure Reachability Condition}
\label{subsec:totalAndPartial}
We present the least fixed point (LFP) semantics with Bellman operators, which are predicate transformers on quantitative predicates for probabilistic systems.
The Bellman operators for total and partial expected rewards on MCs are $\omega$-continuous. Thus, they have least fixed points that are the limits of their respective Kleene iterations.

The Bellman operator $\Psi\colon [S, \interval{\infty}]\to [S, \interval{\infty}]$ for the total expected reward is given by
\begin{align*}
  \Psi(k)(s) \coloneqq \rew(s) + \sum_{s'\in S} P(s, s')\cdot k(s'),
\end{align*}
for each $k\in [S, \interval{\infty}]$ and $s\in S$.
Recall that $[S, \interval{\infty}]$ is the complete lattice of functions $k\colon S\rightarrow \interval{\infty}$ with the pointwise order w.r.t.~the  complete lattice $\interval{\infty}$ of non-negative extended reals.
The least fixed point $\mu \Psi$ is then the assignment of the total expected rewards for each state.
We note that the total expected reward may diverge to $\infty$ even if the set $S$ of states is finite in general.

The definition of the Bellman operator $\Phi$ for partial expected rewards is more involved. In this case, it is needed to calculate reachability probability as well, as the calculation of partial expected rewards depends on it.
Concretely, let $\interval{1} \times \interval{\infty}$ be the product of two complete lattices $\interval{1}$ and $\interval{\infty}$ with the product order. The Bellman operator $\Phi\colon [S, \interval{1} \times \interval{\infty}]\rightarrow [S, \interval{1} \times \interval{\infty}]$ for partial expected rewards is given by
$\Phi(k)(s)\coloneqq (p, r)$ where
\begin{align*}
  p &\coloneqq  P(s, \checkmark) + \sum_{s'\in S} P(s, s')\cdot \pi_1 \big(k(s')\big),\\
  r &\coloneqq  \rew(s) \cdot P(s, \checkmark)  \\
  &\phantom{\coloneqq} + \sum_{s'\in S} P(s, s')\cdot \big(\pi_2 (k(s')) + \rew(s) \cdot \pi_1 (k(s'))\big).
\end{align*}
Note that the computation of the expected reward $r$ depends on $\pi_1 \circ k$, which intuitively captures reachability probabilities.
Given a state $s\in S$, the least fixed point $(\mu \Phi)(s) = (p, r)$ is the pair of the reachability probability $p$ from $s$ to $\checkmark$,
and the partial expected reward $r$ from $s$ to $\checkmark$.
Note that the partial expected reward (on an MC) is finite if $S$ is finite, while it may diverge to $\infty$ if $S$ is infinite.

While partial and total expected rewards are two different semantics that have different applications, they do coincide under some conditions.
One example of such a condition is \emph{almost sure reachability}, which requires that the reachability probability from any state to the target state $\checkmark$ is $1$.
This is simply because the partial expected reward only cares about the paths that reach $\checkmark$, the probability of which is $1$ under the condition, concluding that almost surely all paths are counted.
Specifically, under the almost sure reachability condition, the following two equalities hold:
\begin{align*}
  &\pi_2 \circ \mu \Phi = \mu \Psi, && \mu \Phi =  \langle \Delta_1, \id \rangle\circ \mu \Psi,
\end{align*}
where $\Delta_1\colon \interval{\infty}\rightarrow \interval{1}$ is the constant function $\Delta_1(r) = 1$.

\subsection{A Unifying Theory for the Initial Algebra Correspondence under Reachability Conditions}
\label{subsec:Unify}
We saw a correspondence of total and partial expected rewards on MCs under the almost sure reachability condition above,
and such correspondences of different semantics under appropriate reachability conditions arise naturally in many cases, such as termination conditions of transition systems, but also, for instance, as an unambiguity condition of finite automata.
The aim of this paper is to provide a categorical
foundation of such correspondence under reachability conditions,
providing a general recipe to obtain an appropriate reachability condition for the two different semantics we consider.

We build our framework with initial algebra semantics, that is, the semantics of systems are given as initial algebras of suitable functors.
The least fixed point semantics of predicate transformers are indeed a special case of this initial algebra semantics, where the underlying category is a complete lattice.
Specifically, we aim to answer the following two questions:

\begin{mdframed}
  How does the initial algebra correspondence arise under reachability conditions, and
  how can we obtain such reachability conditions?
\end{mdframed}

The key ingredient in our framework is the adjunction between two domains of endofunctors $\Phi$ and $\Psi$.
Indeed, such adjunctions are fundamental in the context of abstract interpretation~\cite{cousot21}, and coalgebraic semantics (e.g.~\cite{DBLP:journals/logcom/RotJL21, DBLP:journals/iandc/HermidaJ98, DBLP:journals/jcss/Jacobs0S15,DBLP:conf/fossacs/BonsangueK05,DBLP:conf/amast/PavlovicMW06,DBLP:journals/entcs/Klin07}):
we develop our theory on top of this adjunction-based approach.
We then show that the adjunction is liftable to the level of algebras, and address the first question:
the initial algebra correspondence arises by the adjoint equivalence on suitable full subcategories of the categories of algebras, where the adjunction is a lifting of the one between $\Phi$ and $\Psi$.
We encode reachability conditions as \gra{}s, which requires that the initial algebras---representing the system's semantics---lie in these full subcategories,
which answers the second question.
When this condition is satisfied, the adjoint equivalence immediately yields the desired initial algebra correspondence.
Our abstract theory covers several examples that can be found in the literature, where reachability conditions are systematically induced by the underlying adjunctions.

We conclude this overview with the structure of this paper.
In~\S\ref{sec:cat}, we present the background of our theory, namely the adjoint equivalence on the suitable full subcategories of algebras.
In~\S\ref{sec:init_alg_corresp}, we introduce our \gra{}s based on the adjoint equivalence, and demonstrate \gra{}s with concrete examples.
In~\S\ref{sec:init_alg_corresp_L}, we extend our framework for a wider class of predicate transformers, focusing on initial chains in addition for the initial algebra correspondence.
We show that the unambiguity condition for automata is a \gra{} in the extended framework.
In~\S\ref{sec:kanExtension}, we study a systematic way to construct predicate transformers that fit into our framework with pointwise Kan extensions.
We instantiate our approach for Markov decision processes (MDPs).
In~\S\ref{sec:conclusion}, we review our contributions, and list future work.

\section{Equivalence between subcategories of algebras} \label{sec:cat}
After recalling an adjoint equivalence on subcategories for adjunctions (in \S\ref{subsec:adjointEquivalence}),
we show that the equivalence can be lifted to subcategories of algebras (in \S\ref{subsec:equiv_alg}).

\subsection{Background: Fixed Points of Adjoint Functors}
\label{subsec:adjointEquivalence}
An \emph{adjoint equivalence} is an adjunction
$\adj L R \eta \epsilon \colon {\mathbb{C}} \to {\mathbb{D}}$
such that both $\eta$ and $\epsilon$ are natural isomorphisms.
In this subsection,
we recall
folklore results for
inducing adjoint equivalences
on certain subcategories via adjoint functors.
For additional details, see e.g.~\cite{porst1991concrete}.
\begin{definition}
  Let $L \dashv R\colon \mathbb{D} \to \mathbb{C}$ be an adjunction.
  The category \emph{$\fixunit{\mathbb{C}}$}
  is the full subcategory of $\mathbb{C}$
  with objects $X$ such that $\eta_X$  are isomorphisms.
  Similarly,
  the category \emph{$\fixcounit{\mathbb{D}}$}
  is the full subcategory of $\mathbb{D}$
  with objects $X$ such that $\epsilon_X$  are isomorphisms.
\end{definition}

We remark that there is a simple characterization for the equality $\mathbb{D} = \fixunit{\mathbb{D}}$ (and $\mathbb{C} = \fixunit{\mathbb{C}}$ as well).
\begin{proposition} \label{prop:adj_full_faith}
  Let $L \dashv R\colon \mathbb{D} \to \mathbb{C}$ be an adjunction.
  The right adjoint $R\colon \mathbb{D} \to \mathbb{C}$
  (resp.~the left adjoint $L\colon \mathbb{C} \to \mathbb{D}$)
  is full and faithful if and only if the counit $\epsilon$ (resp.~the unit $\eta$) is a natural isomorphism.
  \qed
\end{proposition}

\newcommand{\rs}{(\mathbf{2} \times \mathbf{2})_l}
\begin{example}
  \label{eg:adjointPT}
  Let
  $\rs$ be the complete lattice $\{\bot, \top\} \times \{\bot, \top\}$ equipped
  with the lexicographic order: $(a, b) \sqsubseteq (c, d)$ if and only if $a \leq c$ and $a = c$ implies $b \leq d$.
  We define
  an adjunction
  $L \dashv R\colon \mathbf{2} \to \rs$
  by
  $L(a, b) \coloneqq a \land b$ for each $(a, b) \in \rs$ and
  $R(t)\coloneqq (\top, t)$ for each $t \in \mathbf{2}$.

  With this adjunction,
  the category $\fixunit{(\rs)}$ consists of
  all pairs $(a, b)$ of Booleans  satisfying
  $a = \top$.
  Since the given adjunction is a Galois insertion (i.e.~a Galois connection $L \dashv R$ with $L \circ R = \id$),
  it follows from Prop.~\ref{prop:adj_full_faith} that
  $\fixcounit{\mathbf{2}} = \mathbf{2}$.
\end{example}

We now arrive at a statement about the adjoint equivalence on subcategories $\fixunit{\mathbb{C}}$ and $\fixcounit{\mathbb{D}}$.
\begin{proposition} \label{prop:adj_equiv}
    Any adjunction $\adj L R \eta \epsilon \colon {\mathbb{C}} \to {\mathbb{D}}$ restricts
	to an adjoint equivalence $\adj L R \eta \epsilon \colon {\fixunit{\mathbb{C}}} \to {\fixcounit{\mathbb{D}}}$.
\qed
\end{proposition}

\subsection{Equivalence of subcategories of algebras}
\label{subsec:equiv_alg}
We extend the discussion to categories of algebras.
These categories play a key role in capturing least fixed points in a categorical framework.
\begin{definition} %
  For an endofunctor $\Phi\colon \mathbb{C} \to \mathbb{C}$,
  \emph{the category $\Alg{\Phi}$ of $\Phi$-algebras}
   is defined as follows:
  The objects are
  pairs
  of an object
  $X \in \mathbb{C}$ and a morphism $a\colon \Phi X \to X \text{ in }\mathbb{C}$,
  and the morphisms $(X, a) \to (Y, b)$ are
  $f\colon X \to Y$ in $\mathbb{C}$ satisfying $f \circ a = b \circ \Phi f$.
  When $X$ is clear from the context, we simply write $a$ for the pair $(X, a)$.
\end{definition}
In the definition above, we often instantiate $\mathbb{C}$ to a lattice of predicates, and call $\Phi$ a \emph{predicate transformer.}

Throughout this section,
we use the following running example
that illustrates
partial and total correctness. 
\begin{example}
  \label{eg:ptPredTran}
Consider Example~\ref{eg:adjointPT}.
Let $\Sigma$ be a set of labels,
and let
$c\colon S \to (S+\{\checkmark\})^\Sigma$ be a function representing a deterministic labeled transition system with a terminal state $\checkmark$.
We write
$\Sigma^\omega$ for the set of infinite sequences of $\Sigma$ and
let $A \subseteq S$ be a set of \emph{safe} states.
The target functor $\Phi_c\colon [S \times \Sigma^\omega, \rs]\rightarrow [S \times \Sigma^\omega, \rs]$ is defined by
\begin{align*}
&\Phi_c(k)(s, \sigma \vec{\sigma}) \\
&=
\begin{cases}
  (\top, A(s)) &\text{if }c(s)(\sigma) = \checkmark \\
  (\pi_1 k(s', \vec{\sigma}), \neg \pi_1 k(s', \vec{\sigma}) \lor \pi_2 k(s', \vec{\sigma})) &\text{otherwise},
\end{cases}
\end{align*}
where
$A(s) = \top$ if $s \in A$ and $\bot$ otherwise,
and
$s' = c(s)(\sigma)$.
We use the functor $\Phi_c$ to capture \emph{partial correctness} of the system $c$,
in the sense that if it terminates, it does so safely (i.e.,~in a state of $A$).
For each $n \in \mathbb{N}$ with $n \geq 1$,
letting $\Delta_{(\bot, \bot)}$ be the constant function  returning $(\bot, \bot)$,
the first component of $\Phi_c^n(\Delta_{(\bot, \bot)})(s, v)$
indicates whether
the path determined by $v$ from $s$ terminates in at most $n$-steps,
and the second component indicates whether
it terminates safely (i.e.,~in a state of $A$) if it does terminate in at most $n$-steps.
Since $[S \times \Sigma^\omega, \rs]$ is a complete lattice and $\Phi_c$ is $\omega$-continuous,
the least fixed point $\mu \Phi_c$ can be written as $\bigsqcup_{n \in \mathbb{N}}\Phi_c^n (\Delta_{(\bot, \bot)})$
by Kleene theorem~(see e.g.~\cite{DBLP:journals/dm/Baranga91}).
Specifically,
for each $s \in S$ and $v \in \Sigma^\omega$,
$\mu \Phi_c(s)(v)$
is
the pair of Booleans
that indicates
1) whether the path determined by $v$ from $s$ eventually terminates,
and 2) whether the path terminates safely if it does eventually terminate.

Note that the category $[S \times \Sigma^\omega, \rs]$ also forms a complete lattice,
and thus
the category $\Alg{\Phi_c}$ consists of  all prefixed points of $\Phi_c$.
By the Knaster--Tarski theorem~\cite{pjm/1103044538},
the least fixed point $\mu{\Phi_c}$ of $\Phi_c$ gives
the initial object of $\Alg{\Phi_c}$.
\end{example}

We move to subcategories of algebras defined as follows.
\begin{definition} %
  Consider a subcategory $\mathbb{C}'$ of $\mathbb{C}$.
  For a functor $\Phi\colon \mathbb{C} \to \mathbb{C}$,
  we define the category
  \emph{$\AlgRst{\Phi}{\mathbb{C}'}$}
  by the following change-of-base situation~\cite{DBLP:books/daglib/0023251}.
  \begin{displaymath}
    \xymatrix@R=1em@C=1em{
      \AlgRst{\Phi}{\mathbb{C}'} \ar[d] \pullbackmark{1, 0}{0, 1} \ar[r] &\Alg{\Phi} \ar[d] \\
      \mathbb{C}' \ar@{^{(}->}[r] &\mathbb{C}
    }
  \end{displaymath}
\end{definition}
\begin{proposition} \label{prop:alg_fullsub}
  If $\mathbb{C}'$ is a full subcategory of $\mathbb{C}$,
  then $\AlgRst{\Phi}{\mathbb{C}'}$ is a full subcategory of $\Alg{\Phi}$.
  \qed
\end{proposition}
  The next definition describes how to lift
  functors along the forgetful functor $\AlgRst{\Phi}{\mathbb{C}'} \to \mathbb{C}'$.
  Later,
  in Thm.~\ref{thm:hermida_sub_equiv},
  we will extend this construction further to obtain adjunctions and adjoint equivalences between subcategories of algebras.
\begin{definition} \label{def:alg_functor}
Assume we have two subcategories
  $\mathbb{C}' \hookrightarrow \mathbb{C}$ and $\mathbb{D}' \hookrightarrow \mathbb{D}$ and two functors $\Phi\colon \mathbb{C} \rightarrow \mathbb{C}$ and $\Psi\colon \mathbb{D} \rightarrow \mathbb{D}$.
  Let  $L\colon \mathbb{C} \rightarrow \mathbb{D}$ be a functor that can be restricted to the subcategories $(L\colon \mathbb{C}'\rightarrow  \mathbb{D}')$, and
  $\alpha$ be a natural transformation $\alpha\colon  \Psi L\Rightarrow L\Phi$ as follows:
  \begin{equation} \label{eq:alpha}
    \xymatrix{
      \mathbb{C}' \ar[d]_L \ar[r]^\Phi &\mathbb{C} \ar[d]^L \\
      \mathbb{D}' \ar@{=>}[ur]^\alpha \ar[r]^\Psi &\mathbb{D}
    }
  \end{equation}
  We define a functor $\AlgFn{L}{\alpha}\colon \AlgRst{\Phi}{\mathbb{C}'} \to \AlgRst{\Psi}{\mathbb{D}'}$ by
  \begin{displaymath}
    (X, a) \mapsto (LX, La \circ \alpha_X) \
    \text{ and } \
    f \mapsto Lf.
  \end{displaymath}
  We omit $\alpha$ in $\AlgFn{L}{\alpha}$ when $\alpha$ is given by $\Psi L = L\Phi$.
\end{definition}

\begin{theorem} \label{thm:hermida_sub_equiv}
Consider subcategories, functors, and a natural transformation described in \eqref{eq:alpha},
and assume that $\alpha$ is an isomorphism and there is an adjunction $\adj L R {\eta}{\epsilon}\colon \mathbb{C} \to \mathbb{D}$ that can be restricted to the subcategories.
Let $\beta$ be
the natural transformation $R\Psi \epsilon \circ \overline{\alpha^{-1}}R\colon \Phi R \Rightarrow R\Psi$.

Then the following tuple is an adjunction.
\begin{displaymath}
  \adj {\AlgFn{L}{\alpha}}{\AlgFn{R}{\beta}}{\eta}{\epsilon}\colon \AlgRst{\Phi}{\mathbb{C}'} \to \AlgRst{\Psi}{\mathbb{D}'}.
\end{displaymath}
Moreover, when $\mathbb{C}' = \fixunit{\mathbb{C}}$ and $\mathbb{D}' = \fixcounit{\mathbb{D}}$,
this adjunction becomes an adjoint equivalence.
\qed
\end{theorem}
See 
\ifarxiv
Appendix~\ref{ap:omitted_cat} 
\else
\cite[Appendix A]{arxiv}
\fi
for its proof.
We note that the construction of such adjunctions between subcategories of algebras can be seen as
a natural extension of~\cite[Thm.~2.14]{DBLP:journals/iandc/HermidaJ98}.

In the setting of Thm.~\ref{thm:hermida_sub_equiv}, one finds that
$\Psi \cong L \Phi R$ via $\alpha_R \circ \Psi \epsilon^{-1}$
when
R is full and faithful.
Beyond this, we establish an additional
result about equivalence.
Specifically,
letting $\Psi \coloneqq L\Phi R$ and $\alpha \coloneqq L\Phi \eta^{-1}$,
we obtain the following corollary immediately,
which serves as our main result for the initial algebra correspondences discussed in Sec.~\ref{sec:init_alg_corresp} and~\ref{sec:init_alg_corresp_L}.

\begin{corollary} \label{cor:fix}
  Let $\Phi\colon \mathbb{C} \to \mathbb{C}$ be a functor and
  let
  $\adj L R \eta \epsilon \colon {\mathbb{C}} \to {\mathbb{D}}$ be an adjunction.
  Then the following tuple
 is an adjoint equivalence:
  \begin{equation}\label{eq:alg_equiv}
  \adj {\AlgFn{L}{\alpha}} {\AlgFn{R}{\beta}} \eta \epsilon \colon {\AlgRst{\Phi}{\fixunit{\mathbb{C}}}} \to {\AlgRst{L\Phi R}{\fixcounit{\mathbb{D}}}},
  \end{equation}
  where $\alpha = L\Phi \eta^{-1}$ and $\beta = \eta \Phi R$.
  \qed
\end{corollary}

\begin{example} \label{eg:ptPredTran2}
Consider Example~\ref{eg:ptPredTran}.
The adjunction defined in Example~\ref{eg:adjointPT}
induces, in a pointwise manner, an adjunction between
  functor categories $[S \times \Sigma^\omega, \rs]$ and $[S \times \Sigma^\omega, \mathbf{2}]$.
  We denote this adjunction by $L \dashv R\colon [S \times \Sigma^\omega, \mathbf{2}] \to [S \times \Sigma^\omega, \rs]$.
Then the composite functor $L\Phi_c R$ is concretely given by
\begin{align*}
  (L \Phi_c R)(k)(s, \sigma\vec{\sigma})
  =
  \begin{cases}
    A(s) &\text{if }c(s)(\sigma) = \checkmark \\
    \pi_2 k(c(s)(\sigma), \vec{\sigma}) &\text{otherwise}.
  \end{cases}
\end{align*}
Since the counit is a natural isomorphism, the equality $\AlgRst{L\Phi_c R}{\fixcounit{[S \times \Sigma^\omega,  \mathbf{2}]}} = \Alg{L\Phi_c R}$ holds, and
the domain $\mu (L \Phi_c R)$ of the initial object
is
the least fixed point of $L \Phi_c R$.
Concretely, $\mu (L \Phi_c R)$
sends each pair $(s, v) \in S \times \Sigma^\omega$ to
the Boolean indicating \emph{total correctness} of $c$,
i.e.~whether the path determined by $v$ from $s$ terminates safely.

By Cor.~\ref{cor:fix},
we obtain
an adjoint equivalence between $\AlgRst{\Phi_c}{\fixunit{[S \times \Sigma^\omega, \rs]}}$
and $\Alg{L\Phi_c R}$.
Under this equivalence,
a prefixed point $k$ of $\Phi_c$ lies in $\fixunit{[S \times \Sigma^\omega, \rs]}$ corresponds to a prefixed point of $L\Phi_c R$ via $L$ and $R$.
In the next section, we show that this equivalence maps initial algebras under an abstract reachability condition.
\end{example}

\subsection{Extension for liftings of adjunctions} \label{subsec:lifting}
We briefly mention an extension of Cor.~\ref{cor:fix} to the setting of liftings of adjunctions rather than merely adjunctions
(in other words, adjunctions in the arrow 2-category $\CAT^{\to}$ rather than adjunctions in the 2-category $\CAT$ of categories).
This extension offers a framework for analyzing how equivalences occur in fibrational settings.
The result established here is used only in the example in \S{}\ref{subsec:corresp_lifting};
the reader may skip to the next section and revisit this part when exploring \S{}\ref{subsec:corresp_lifting}.

For functors $p\colon \mathbb{E} \to \mathbb{C}$, $q\colon \mathbb{F} \to \mathbb{D}$, and $F\colon \mathbb{C} \to \mathbb{D}$,
we say that a functor $\dot{F}\colon \mathbb{E} \to \mathbb{F}$ is a \emph{lifting of $F$ along $p, q$}
if $F \circ p = q \circ \dot{F}$.
Similarly, for liftings $\dot{F}, \dot{G}$ of $F, G$ (respectively) along $p, q$, and
a natural transformation $\alpha\colon F \Rightarrow G$,
we say that a natural transformation $\dot{\alpha}\colon \dot{F} \Rightarrow \dot{G}$ is a \emph{lifting of $\alpha$ along $p, q$} if $\alpha p = q \dot{\alpha}$.
We often omit $q$ in ``along $p, q$'' when $p = q$.

A \emph{lifting of an adjunction (resp.~adjoint equivalence) $\adj L R \eta \epsilon\colon {\mathbb{C}} \to {\mathbb{D}}$}
along $p\colon \mathbb{E} \to \mathbb{C}, q\colon \mathbb{F} \to \mathbb{D}$
is an adjunction (resp.~adjoint equivalence) $\adj {\dot{L}} {\dot{R}} {\dot{\eta}} {\dot{\epsilon}} \colon {\mathbb{E}} \to {\mathbb{F}}$
s.t.\
$\dot{L}, \dot{R}, \dot{\eta}, \dot{\epsilon}$ are liftings of
$L, R, \eta, \epsilon$, respectively, along $p, q$.

\begin{proposition} \label{prop:fix_total}
  Let $p\colon \mathbb{E} \to \mathbb{C}$ and $q\colon \mathbb{F} \to \mathbb{D}$ be functors.
  Suppose $\dot{\Phi}\colon \mathbb{E} \to \mathbb{E}$ is a lifting of $\Phi\colon \mathbb{C} \to \mathbb{C}$ along $p$,
  and
  an adjunction
  $\adj {\dot{L}} {\dot{R}} {\dot{\eta}} {\dot{\epsilon}} \colon {\mathbb{E}} \to {\mathbb{F}}$ is a lifting of an adjunction $\adj L R \eta \epsilon \colon {\mathbb{C}} \to {\mathbb{D}}$ along $p, q$.

By applying Cor.~\ref{cor:fix} to the pairs $(\dot{\Phi}, \adj {\dot{L}} {\dot{R}} {\dot{\eta}} {\dot{\epsilon}})$ and $(\Phi, \adj L R \eta \epsilon)$, respectively,
we obtain
the following adjoint equivalence
\begin{center}
  $\adj {\AlgFn{\dot{L}}{\dot{\alpha}}} {\AlgFn{\dot{R}}{\dot{\beta}}} {\dot{\eta}} {\dot{\epsilon}} \colon {\AlgRst{\dot{\Phi}}{\fixunit{\mathbb{E}}}} \to {\AlgRst{\dot{L}\dot{\Phi} \dot{R}}{\fixcounit{\mathbb{F}}}}$,
\end{center}
which
is a lifting of
the adjoint equivalence
\begin{center}
$\adj {\AlgFn{L}{\alpha}} {\AlgFn{R}{\beta}} \eta \epsilon \colon {\AlgRst{{\Phi}}{\fixunit{\mathbb{C}}}} \to {\AlgRst{{L}{\Phi} {R}}{\fixcounit{\mathbb{D}}}}$
\end{center}
along $\Alg{p}, \Alg{q}$
where $\dot{\alpha} = \dot{L}\dot{\Phi}\dot{\eta}^{-1}$, $\dot{\beta} = \dot{\eta}\dot{\Phi}\dot{R}$,
${\alpha} = {L}{\Phi}{\eta}^{-1}$, and ${\beta} = {\eta}{\Phi}{R}$.
  \qed
\end{proposition}

\section{Initial algebra correspondence by adjoint equivalences} \label{sec:init_alg_corresp}
Based on the adjoint equivalence between subcategories of algebras shown in
Cor.~\ref{cor:fix}, we introduce a framework for deriving initial algebra correspondence under a reachability condition.
Throughout this section,
we work within the following setting:
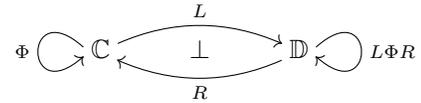
\begin{displaymath}
  \xymatrix{
    \mathbb{C} \lloop{\Phi} \ar@/^1em/[rr]^{L} &\bot &\mathbb{D} \ar@/^1em/[ll]^{R} \rloop{L \Phi R}
  }
\end{displaymath}
We assume that
both $\Phi$ and $L\Phi R$ have initial algebras, denoted $(\mu \Phi, \iota_\Phi)$ and $(\mu (L\Phi R), \iota_{L \Phi R})$, respectively.
Our goal here is to establish a sufficient condition ensuring
a correspondence between $\iota_\Phi$ and $\iota_{L \Phi R}$
via the equivalence \eqref{eq:alg_equiv}.
\begin{definition}
The \emph{initial algebra correspondence for $\Phi$ and $L \Phi R$} holds if
the following condition $(\star)$ is satisfied.
\begin{align*}
  (\star): \quad& \iota_\Phi\text{ is initial in }\AlgRst{\Phi}{\fixunit{\mathbb{C}}},\\
  & \text{and }\iota_{L \Phi R}\text{ is initial in } \AlgRst{L\Phi R}{\fixcounit{\mathbb{D}}}.
\end{align*}
\end{definition}
\begin{proposition}
The \emph{initial algebra correspondence for $\Phi$ and $L \Phi R$} holds iff
$\iota_\Phi$ and $\iota_{L \Phi R}$ lie in $\AlgRst{\Phi}{\fixunit{\mathbb{C}}}$ and $\AlgRst{L\Phi R}{\fixcounit{\mathbb{D}}}$, respectively. They correspond to each other by the adjoints $\AlgFn{L}{\alpha}$ and $\AlgFn{R}{\beta}$ given in Cor.~\ref{cor:fix}.
\end{proposition}
\begin{proof}
  The (only-if) is clear.
  For (if),
by Prop.~\ref{prop:alg_fullsub},
$\AlgRst{\Phi}{\fixunit{\mathbb{C}}}$ and $\AlgRst{L\Phi R}{\fixcounit{\mathbb{D}}}$ are full subcategories of $\Alg{\Phi}$ and $\Alg{L \Phi R}$, respectively.
Thus $\iota_\Phi$ lies in $\AlgRst{\Phi}{\fixunit{\mathbb{C}}}$
if and only if $\iota_\Phi$ is initial in $\AlgRst{\Phi}{\fixunit{\mathbb{C}}}$,
and the same statement for $\iota_{L \Phi R}$ and $\AlgRst{L\Phi R}{\fixcounit{\mathbb{D}}}$ holds.
\end{proof}
Furthermore, the initial algebra correspondence automatically induces the \emph{least fixed-point (lfp) correspondence}, namely,
$\mu \Phi \cong R\big(\mu (L \Phi R)\big)$ and $L(\mu \Phi) \cong \mu (L \Phi R)$.

We particularly focus on the case that $R$ is full and faithful.
This condition aligns with the concept of Galois insertions,
a fundamental notion in the theory of abstract interpretation~\cite{DBLP:conf/popl/CousotC77}.
If $R$ is full and faithful,
then by Prop.~\ref{prop:adj_full_faith}
we immediately have
$\fixcounit{\mathbb{D}} = \mathbb{D}$,
and hence
the second condition of $(\star)$ is automatically satisfied.
As a result, to establish the initial algebra correspondence,
it suffices to verify only the first condition.
To this end, we impose
a reachability condition defined below, which implies the first condition of $(\star)$.

\begin{definition}[\gra]
  A $\Phi$-algebra $a$
  satisfies the
  \emph{\gra}
  if $a \in \AlgRst{\Phi}{\fixunit{\mathbb{C}}}$.
  Equivalently, the codomain of $a$ is in $\fixunit{\mathbb{C}}$.
\end{definition}

\begin{proposition} \label{prop:gra_correspondence}
  Assume $R$ is full and faithful.
  Under the \gra of $\iota_\Phi$,
  the initial algebra correspondence for $\Phi$ and $L\Phi R$ holds.
  \qed
\end{proposition}
\begin{example}[total and partial correctness] \label{eg:total_partial_correct}
  Consider Example~\ref{eg:ptPredTran}.
  The \gra of $\iota_{\Phi_c}$
  is $\mu \Phi_c \in \fixunit{[S \times \Sigma^\omega, \rs]}$,
  that is,
  for each state $s \in S$ and word $v \in \Sigma^\omega$,
  the path determined by $v$ from $s$ eventually terminates.

  Under this condition,
  Prop.~\ref{prop:gra_correspondence} ensures the initial algebra correspondence for $\Phi_c$ and $L \Phi_c R$,
  and thus also guarantees the lfp correspondence
  between
  $\mu \Phi_c$
  and
  $\mu (L \Phi_c R)$.
  It
  captures the correspondence between partial correctness and total correctness
  when $c$ eventually terminates for each input.
\end{example}

\begin{remark}[completeness of an abstract domain] \label{rm:complete}
  When we regard complete lattices as categories,
the notion of \emph{(backward) completeness of an abstract domain} in abstract interpretation~\cite{DBLP:conf/popl/CousotC79, DBLP:journals/jacm/GiacobazziRS00} can be characterized
by the existence of a natural isomorphism
$\alpha\colon (L\Phi R)L \Rightarrow L\Phi$ between functors $\mathbb{C} \Rightarrow \mathbb{D}$.
The existence of $\alpha$ guarantees
a left adjoint $\Alg{L}\colon \Alg{\Phi } \to \Alg{L\Phi R}$ by~\cite[Thm.~2.14]{DBLP:journals/iandc/HermidaJ98},
implying $L (\mu \Phi) \cong \mu (L \Phi R)$.
If we have such an isomorphism $\alpha$,
then
the lfp correspondence
$\mu \Phi \cong R\big(\mu (L \Phi R)\big)$
under the setting of
Prop.~\ref{prop:gra_correspondence}
can be shown
without using Cor.~\ref{cor:fix}.
Concretely,
it follows that
$\mu \Phi \cong RL (\mu \Phi) \cong R \big(\mu(L \Phi R)\big)$
where the first isomorphism is given by the unit of $\mu \Phi$
and the second comes from
preservation of initial objects by the left adjoint $\Alg{L}$.

We provide examples satisfying completeness of an abstract domain in Example~\ref{eg:total_partial_correct} and~\ref{eg:cost-boundedReach}, which also fit into our framework.
By contrast, all other examples in this section are not complete.
\end{remark}

\subsection{Examples of Initial Algebra Correspondence}
We illustrate \gra{}s that induce the initial algebra correspondence
by examining two examples.
Like in Example~\ref{eg:total_partial_correct},
the functor $\Phi$ is implicitly parameterized by a coalgebra $c$,
and a \gra corresponds to a reachability condition imposed on $c$.

\begin{example}[resource-bounded reachability]
  \label{eg:cost-boundedReach}
We fix $M \in \mathbb{N}$.
Consider a coalgebra $c\colon S \to \mathcal{P}(S) \times \mathbb{N} + \{\checkmark\}$ in $\Set$,
representing a non-deterministic transition system  with resources and the target state $\checkmark$.
We are interested in whether each state $s\in S$ can reach the target state $\checkmark$ via a path whose  cumulative sum of resources is
at least $M$; we refer to this as
the \emph{resource-bounded reachability}.
We instantiate our framework for operators for resource-bounded reachability
and ordinary reachability. Such resource-bounded reachability has been actively studied in probabilistic verification, e.g.~\cite{stella1998optimization,ChristmanC13,HartmannsJKQ20}.

  Let $M$ be the complete lattice $\{0, \dots, M\}$ with the standard order, and
  let $M_\bot$ be the pointed complete lattice defined by adding a least element $\bot$.
  Define
  $\Phi\colon [S, M_\bot]\rightarrow [S, M_\bot]$ by
 \[
  \Phi(k)(s) \coloneqq
  \begin{cases}
   0 &\text{if }c(s) = \checkmark, \\
   \bigsqcup_{m \in T}\max(M, m+n)  &\text{if }c(s) \in \mathcal{P}(S) \times \mathbb{N},
  \end{cases}
\]
for each $k\in [S, M_\bot]$ and $s\in S$, where
$(X, n) \coloneqq c(s)$ and
$T\coloneqq k(X)\cap \mathbb{N}$.
Then the initial algebra $(\mu \Phi, \iota_\Phi)$ of $\Phi$ captures resource-bounded reachability:
$\mu \Phi(s) = \bot$ if $s$ is unreachable,
$\mu \Phi(s) = n$ (where $n < M$) if $s$ is reachable with resource value $n$, and
$\mu \Phi(s) = M$ if $s$ is reachable with resource value $M$ or more.
Hence, we can check resource-bounded reachability by checking whether the equality $\mu \Phi(s) = M$ holds.

We show that, under a \gra, resource-bounded reachability can be reduced to the simple reachability objective,
which we encode via $L \Phi R$ as follows.
Consider
the adjunction
$L \dashv R\colon [S, \mathbf{2}] \to [S, M_\bot]$ defined by
(i) $L(k)(s) \coloneqq \bot$ if $k(s) = \bot$, and $\top$ otherwise; and (ii) $R(k)(s) \coloneqq \bot$ if $k(s) = \bot$, and $M$ otherwise.
Then
the functor $L \Phi R\colon [S, \mathbf{2}] \rightarrow [S, \mathbf{2}]$ is given by
\begin{displaymath}
(L \Phi R)(k)(s) = \begin{cases}
  \top &\text{if }c(s) = \checkmark, \\
  \bigsqcup_{t \in \pi_1(c(s))} t  &\text{if } c(s) \in \mathcal{P}(S) \times \mathbb{N}.
  \end{cases}
\end{displaymath}
Then $\mu (L \Phi R)$ represents the reachability to $\checkmark$.

The \gra $(\mu\Phi \in \fixunit{[S, M_\bot]})$ is that all reachable states have paths whose accumulated resources are at least $M$.
Under this condition,
Prop.~\ref{prop:gra_correspondence} ensures
the initial algebra correspondence for $\Phi$ and $L \Phi R$.
As a result,
verifying resource-bounded reachability
can be reduced to
a simple reachability check.
\end{example}

\begin{example}[total and partial expected rewards of MCs] \label{eg:total_partial_mc}
Recall that an MC is
  a coalgebra $c \coloneqq \langle P, \rew\rangle\colon S\rightarrow \mcfunc{S} \times \mathbb{N}$ in $\Set$, where
$P\colon S\to \mathcal{D}(S + \{\checkmark\})$
describes the transition probabilities with finite supports, and $\rew\colon S\rightarrow \mathbb{N}$ is the reward function.
We instantiate our framework for the Bellman operators for the partial expected rewards (e.g.~\cite{Baier0KW17}) and the total expected rewards (e.g.~\cite{Puterman94}), which are formally presented in~\S\ref{subsec:totalAndPartial}.

We show that, under a \gra, the partial expected reward coincides with the total expected reward.
To see this, we consider functors $\Phi$ and $\Psi$ defined in~\S\ref{subsec:totalAndPartial}.
They align with our framework since $\Psi$ can be written as $L \Phi R$, constructed as follows.
  We define an adjunction
  $L \dashv R\colon [S, \interval{\infty}] \to [S, \interval{1} \times \interval{\infty}]$
  by $L(k) \coloneqq \pi_2 \circ k$ and
  $R(k)\coloneqq \langle \Delta_1, \id \rangle \circ k$ where $\Delta_1$ is the constant function returning $1$.
Then the functor $L\Phi R$ is concretely given by
\begin{align*}
  (L \Phi R)(k)(s) = \rew(s) + \sum_{s'\in S} P(s, s')\cdot k(s').
\end{align*}
This functor $L \Phi R$ is precisely the Bellman operator for the total reward objective.

The \gra $(\mu \Phi \in \fixunit{[S, \interval{1} \times \interval{\infty}]})$
is the almost sure reachability of the MC $c$, that is, the reachability probability to $\checkmark$ is $1$ from any state.
  Under this condition,
  Prop.~\ref{prop:gra_correspondence} ensures the initial algebra correspondence for $\Phi$ and $L \Phi R$,
  and thus also guarantees the lfp correspondence
  between
  $\mu \Phi$
  and
  $\mu (L \Phi R)$.
  It
  captures the correspondence between partial expected rewards and total expected rewards
  when $c$ is almost surely reachable to $\checkmark$.
\end{example}

\subsection{Example Beyond Initial Algebra Correspondence}
The following example does not precisely align with our framework,
  but nevertheless we show that a \gra still guarantees a correspondence of algebras via the equivalence induced by adjoints.
  Specifically, we consider
the correspondence
between an initial algebra and a particular algebra
that, while not initial in $\Alg{\Phi}$, becomes initial in $\AlgRst{\Phi}{\fixunit{\mathbb{C}}}$ under a \gra.
For omitted proofs, see 
\ifarxiv
Appendix~\ref{ap:omitted_mdp}.
\else
\cite[Appendix C]{arxiv}.
\fi

\begin{example}[total and partial expected rewards of multi-objective MDPs]
\label{eg:achievable_mdp}
We extend Example~\ref{eg:total_partial_mc} for MCs to an example for MDPs with countable states.
In this setting, the non-determinism introduced by schedulers on MDPs
renders the Bellman operator for the partial expected rewards
\emph{multi-objective}: This multi-objectivity comes
from the compositional computation of partial expected rewards requires to compute not only partial expected reward itself but also reachability probability (see Example~\ref{eg:total_partial_mc}),
which gives two objectives to optimize for schedulers.
Such multi-objective Bellman operators
have been studied in probabilistic verification, e.g. in~\cite{ChenFKSW13,AshokCKWW20}.
We instantiate our framework for a multi-objective Bellman operator for partial expected rewards and the Bellman operator for the total expected rewards~\cite{Puterman94}.

We fix the endofunctor $F \coloneqq \mathcal{D}(- + \{\checkmark\}) \times \mathbb{N}$ on $\Set$.
Consider
a coalgebra $c\colon S \to \Pfin FS$ in $\Set$,
where $S$ is a countable set of states and $c(s) \neq \emptyset$ for each $s \in S$.
Here, $\Pfin$ denotes the finite powerset functor,
and
the coalgebra $c$
represents an MDP.
For any complete lattice $L$, we write $L^\downarrow$ for
the complete lattice of
downward closed subsets
of $L$ with the subset order $\subseteq$.
For $x\in L$, we also write $x^{\downarrow}$ for the principal downward closed set $\{x'\in L\mid x'\sqsubseteq x\}$.

We define an endofunctor
$\Phi$ on $[S, (\interval{1} \times \interval{\infty})^\downarrow]$ as the multi-objective Bellman operator for the partial expected reward. Formally, the Bellman operator  is defined by
  \begin{align*}
\Phi(k)(s) \coloneqq
    \bigcup_{\substack{(d, n) \in c(s),
    k'\colon S \to \interval{1} \times \interval{\infty} \\ \text{ s.t.~}\forall s' \in S.~k'(s') \in k(s')}}
    &\Big(\big(p(d, n, k'), r(d, n, k')\big)\Big)^{\downarrow},
  \end{align*}
  where
    $p(d, n, k') \in \interval{1}$ and $r(d, n, k') \in \interval{\infty}$ are defined by
  \begin{align*}
    p(d, n, k') &\coloneqq d(\checkmark) + \sum_{s'\in S} d(s')\cdot \pi_1  \big(k'(s')\big), \\
    r(d, n, k') &\coloneqq
    n \cdot d(\checkmark)  \\
    &\phantom{\coloneqq}+ \sum_{s'\in S} d(s')\cdot (\pi_2  \big(k'(s')\big) + n \cdot \pi_1  \big(k'(s')\big)).
  \end{align*}

The domain $\mu \Phi$ of the initial algebra is the function mapping each $s \in S$ to the emptyset.
However,
our focus is not on $\mu \Phi$
but on a specific object $f \in [S, (\interval{1} \times \interval{\infty})^\downarrow]$
representing
the multi-objective solutions for the reachability objective and the partial expected reward objective.
This $f$ is defined by
  $f(s) \coloneqq \bigcup_{\text{scheduler }\sigma}(\Prb{\sigma}(s, \checkmark), \Erew{\sigma}(s, \checkmark))^\downarrow$,
  where $\Prb{\sigma}(s, \checkmark)$ and $\Erew{\sigma}(s, \checkmark)$
  denote the probability and the expected reward of reaching the target state $\checkmark$ from $s$ under the scheduler $\sigma$, respectively.
  We consider only deterministic schedulers (see 
  \ifarxiv
  Appendix~\ref{sec:defScheduler} 
  \else
  \cite[Appendix B]{arxiv}
  \fi
  for details).
  Note that the object $f$ forms a $\Phi$-algebra $(f, \Phi(f) \sqsubseteq f)$.

  We show that, under a \gra, $f$ coincides with the total expected rewards,
  which we encode via $L\Phi R$ as follows.
  Consider the adjunction $L\dashv R\colon [S, \interval{\infty}] \to [S, (\interval{1} \times \interval{\infty})^\downarrow]$ defined by
  (i) $L(k) \coloneqq \pi_2  \circ \sqcup \circ k$;
  where $\sqcup\colon (\interval{1} \times \interval{\infty})^\downarrow \to \interval{1} \times \interval{\infty}$ is given by the join in $\interval{1} \times \interval{\infty}$;
  and (ii) $R(k)\coloneqq (-)^\downarrow \circ \langle \Delta_1, \id\rangle \circ k$
  where $\Delta_1\colon \interval{\infty} \to \interval{1}$ is the constant function returning $1$.
  Then
  the endofunctor $L \Phi R$ on $[S, \interval{\infty}]$ is the Bellman operator
  given by
  \begin{align*}
    (L \Phi R)(k)(s) = \max_{(d, n)\in c(s)} n + \sum_{s'\in S} d(s')\cdot k(s').
  \end{align*}
  Likewise the case for MCs,
  the domain $\mu (L \Phi R)$ of the initial algebra represents the total expected reward,
  which is the supremum of the expected rewards
  from $s$ under all (deterministic) schedulers $\sigma$.

  The \gra of $(f, \Phi(f) \sqsubseteq f)$
  ensures that $f$ is initial (see 
  \ifarxiv
  Lem.~\ref{lem:f_init} 
  \else
  \cite[C.1]{arxiv}
  \fi
  for the proof).
  This leads to
  a result analogous to Prop.~\ref{prop:gra_correspondence}:
  Under the \gra of $(f, \Phi(f) \sqsubseteq f)$,
  $(f, \Phi f \sqsubseteq f)$ corresponds to $\iota_{L \Phi R}$ via the equivalence~\eqref{eq:alg_equiv}.

A classical reachability condition for ensuring the correspondence between
$f$ and $\mu (L \Phi R)$ is that of almost sure reachability, i.e., whichever scheduler is chosen, $\checkmark$ will be reached with probability $1$ (see e.g.~\cite{DBLP:books/daglib/0020348}).
Here, as the \gra of $(f, \Phi f \sqsubseteq f)$, 
we obtain a different condition 
for MDPs with countable states
as follows:
We have an optimal scheduler $\sigma$ such that $P_\sigma(s, \checkmark) = 1$ for each state $s \in S$.
This condition is weaker than the classical one under the existence of an optimal scheduler.
Indeed, our \gra allows us to have a scheduler that does not reach $\checkmark$ with probability $1$.
\end{example}

\subsection{Initial Algebra Correspondence, Fibrationally} \label{subsec:corresp_lifting}
Through previous examples, including our running example (Example~\ref{eg:total_partial_correct}),
we demonstrated how \gra{}s represent reachability conditions on coalgebras that parameterize $\Phi$,
and explored initial algebra correspondences under these conditions.
In each case, we derived the correspondences via equivalences established by Cor.~\ref{cor:fix} between posetal categories (categories defined by posets).

In this section, we use fibrations and Prop.~\ref{prop:fix_total} 
to move from a single coalgebra to the category of coalgebras,
and derive equivalences between non-posetal categories.
We introduce a fibration that allows us to
represent
prefixed points of $\Phi$ (or $\Phi_c$ in Example~\ref{eg:total_partial_correct})
as coalgebras above $c$ along the fibration.

For a poset $\Omega$,
  we define a category $\fibop{(\Set/\Omega)}$ with
  \begin{itemize}
    \item objects: pairs of a set $X$ and a function $f\colon X \to \Omega$.
    \item morphisms: $(f\colon X \to \Omega) \to (g\colon Y \to \Omega)$ are functions $h\colon X \to Y$
      satisfying $g(h(x)) \sqsubseteq f(x)$ for each $x \in X$.
  \end{itemize}
  The forgetful functor from $\fibop{(\Set/\Omega)}$ to $\Set$ forms a fibration,
  which
  is called the \emph{fibrewise opposite of lax domain fibrations}.
  We write $\fibop{\laxdom{\Omega}}$ for this fibration.
  We refer to~\cite[Def.~1.10.10]{DBLP:books/daglib/0023251} for more details on fibrewise opposite fibrations.

  Given a monotone function $a\colon \Omega \to \Omega'$ between posets,
  there is a functor $a_* \colon \fibop{(\Set/\Omega)} \to \fibop{(\Set/\Omega')}$
  defined by $a_*(f) = a \circ f$ for each object $f$ and $a_*(h) = h$ for each morphism $h$.

  The category $\Coalg{F}$ of $F$-coalgebras
  is defined as $\Alg{F^\op}^\op$.
  Explicitly, the objects are pairs of $X \in \mathbb{C}$ and $c\colon X \to FX \text{ in }\mathbb{C}$,
  and the morphisms $(X, c) \to (Y, d)$ are $f\colon X \to Y$ in $\mathbb{C}$ satisfying $Ff \circ c = d \circ f$.

  If $\dot{F}\colon \mathbb{E} \to \mathbb{E}$ is a lifting of $F\colon \mathbb{C} \to \mathbb{C}$ along a functor $p\colon \mathbb{E} \to \mathbb{C}$,
  we have
  the functor $\Coalg{p}\colon \Coalg{\dot{F}} \to \Coalg{F}$
  defined by $\Alg{p^\op}^\op$.

\begin{example} \label{eg:fibrational}
  Consider Example~\ref{eg:total_partial_mc}.
  Let $\Omega_1 \coloneqq \interval{1} \times \interval{\infty}$, $\Omega_2 \coloneqq \interval{\infty}$,
  and
  $F \coloneqq \mathcal{D}(- + \{\checkmark\}) \times \mathbb{N}$.
  Recall that $\Phi$ is implicitly parameterized by
  a coalgebra $c\colon S \to FS$ representing an MC,
  and
  $\Alg{\Phi}$ is the category of prefixed points of $\Phi$.
  We begin by interpreting $\Alg{\Phi}$ in a fibrewise manner.
  Define $\dot{F}\colon \fibop{(\Set/\Omega_1)} \to \fibop{(\Set/\Omega_1)}$ by
  \begin{displaymath}
   \dot{F}\big((X, k)\big) \coloneqq (FX, k') \text{ and } \dot{F}(h) \coloneqq Fh,
  \end{displaymath}
    where $k'\colon FX \to \Omega_1$ sends each $(d, n) \in FX$ to
    \begin{align*}
  &\left(d(\checkmark) + \Sigma_{x \in X} d(x) \cdot \pi_1 (k(x)), \right. \\
  &\left. \phantom{\big(}n \cdot d(\checkmark) + \Sigma_{x \in X}d(x) \cdot \big(\pi_2 (k(x)) + n \cdot \pi_1 (k(x))\big)\right).
    \end{align*}
  Since $\dot{F}$ is a lifting of $F$ along $\fibop{\laxdom{\Omega_1}}$,
  it induces a functor
  \begin{displaymath}
    \Coalg{\fibop{\laxdom{\Omega_1}}}\colon\Coalg{\dot{F}} \to \Coalg{F}.
  \end{displaymath}
  The fibre category $\Coalg{\dot{F}}_c$ is isomorphic to $\Alg{\Phi}$.

  A similar construction applies to $\Alg{L \Phi R}$.
  Consider
  the adjunction $\pi_2  \dashv \langle \Delta_1, \id \rangle\colon \interval{\infty} \to \interval{1} \times \interval{\infty}$.
  This immediately
  induces
  the adjunction $\langle \Delta_1, \id \rangle_* \dashv (\pi_2)_*\colon \fibop{(\Set/\Omega_1)} \to \fibop{(\Set/\Omega_2)}$,
  which is a lifting of $\id \dashv \id$ along $\fibop{\laxdom{\Omega_1}}$, $\fibop{\laxdom{\Omega_2}}$.
  Hence,
  we obtain a functor
  \begin{displaymath}
  \Coalg{\fibop{\laxdom{\Omega_2}}}\colon\Coalg{(\pi_2)_*\dot{F}\langle \Delta_1, \id \rangle_*} \to \Coalg{F},
  \end{displaymath}
  and
  the fibre category above $c$ is isomorphic to $\Alg{L \Phi R}$.

  By the dual version of Prop.~\ref{prop:fix_total} (see 
  \ifarxiv
  Appendix~\ref{ap:dual_fix_total}),
  \else
  \cite[Appendix G]{arxiv}),
  \fi
  we obtain
   the following diagram:
  \begin{displaymath}
    \xymatrix@R=3em@C=0.3em{
      \CoalgRst{\dot{F}}{\fixcounit{(\Set/\Omega_1)}} \ar[dr]_{\Coalg{\fibop{\laxdom{\Omega_1}}}}  \ar@/^1em/[rr]^-{\CoalgFn{(\pi_2)_*}{\dot{\alpha}}} &\simeq &\Coalg{(\pi_2)_*\dot{F}\langle \Delta_1, \id \rangle_*} \ar@/^1em/[ll]^-{\CoalgFn{\langle \Delta_1, \id \rangle_*}{\dot{\beta}}} \ar[dl]^{\Coalg{\fibop{\laxdom{\Omega_2}}}} \\
      &\Coalg{F}
    }
  \end{displaymath}
  Here,
  $\dot{\alpha} \coloneqq (\pi_2)_*\dot{F}\dot{\epsilon}^{-1}$ and
  $\dot{\beta} \coloneqq \dot{\epsilon} \dot{F}\dot{L}$
  where $\dot{\epsilon}$ is the counit of the adjunction $\langle \Delta_1, \id \rangle_* \dashv (\pi_2)_*$.
  Note that $\fixunit{(\Set/\Omega_2)}$ is equal to $\Set/\Omega_2$ by Prop.~\ref{prop:adj_full_faith}.
  The top horizontal line in the diagram represents the adjoint equivalence
    $\adj  {\CoalgFn{\langle \Delta_1, \id \rangle_*}{\dot{\beta}}} {\CoalgFn{(\pi_2)_*}{\dot{\alpha}}}{\dot{\eta}} {\dot{\epsilon}}$,
    which is a lifting of the identity adjunction $\id \dashv \id$
  along $\Coalg{\fibop{\laxdom{\Omega_2}}}$, $\Coalg{\fibop{\laxdom{\Omega_1}}}$.
  The equivalence between fibres above $c$ coincides with
  $\AlgRst{\Phi}{\fixunit{[S, \interval{1} \times \interval{\infty}]}} \simeq \Alg{L\Phi R}$,
  established by Cor.~\ref{cor:fix}.
In this context, for a coalgebra $c$,
  the \gra discussed in Example~\ref{eg:total_partial_mc} can be translated to:
  The final object in the fibre $\Coalg{\dot{F}}_c$
  lies in
$(\CoalgRst{\dot{F}}{\fixcounit{(\Set/\Omega_1)}})_c$.
\end{example}

For all examples in \S{}\ref{sec:init_alg_corresp}, there are similar extensions to the fibrational setting as in Example~\ref{eg:fibrational} 
(see 
\ifarxiv
Appendix~\ref{ap:fib_lift}
\else
\cite[Appendix D]{arxiv}
\fi
 for further details).
Although in Example~\ref{eg:achievable_mdp}
we restricted $c$ to coalgebras representing countable MDPs,
these equivalences can be extended to arbitrary coalgebras in $\Coalg{\Pfin F}$.

\newcommand{\Ord}{\mathbf{Ord}}
\newcommand{\colim}{\mathrm{colim}}
\section{Initial Algebra Correspondence by Chain Correspondence} \label{sec:init_alg_corresp_L}
As demonstrated in the previous section,
\gra{}s yield initial algebra correspondences for $\Phi$ and $L \Phi R$ via
adjoint equivalences.
In this section,
we investigate the initial algebra correspondence for $\Phi\colon \mathbb{C}\rightarrow \mathbb{C}$ and
$\Psi\colon \mathbb{D}\rightarrow \mathbb{D}$ without restricting to $\Psi \cong L\Phi R$.
We show that the initial algebra correspondence can also be guaranteed under \gra{}s, specifically when such conditions ensure a correspondence of chains used to compute initial algebras.
Concretely, we consider  the following setting:
\begin{displaymath}
  \xymatrix{
    \mathbb{C} \lloop{\Phi} \ar@/^1em/[rr]^{L} &\bot &\mathbb{D} \ar@/^1em/[ll]^{R} \rloop{\Psi}
  }
  \text{ with }\rho\colon L \Phi R \Rightarrow \Psi
\end{displaymath}
This natural transformation $\rho$
is referred to as a \emph{step} in coalgebraic trace semantics~\cite{DBLP:journals/logcom/RotJL21}.
The natural transformation $\rho$ induces a functor $\AlgFn{\id}{\rho}\colon \Alg{\Psi} \to \Alg{L \Phi R}$ by Def.~\ref{def:alg_functor}.
We assume that both $\Phi$ and
$\Psi$ have initial algebras,  and we write $(\mu \Phi, \iota_\Phi)$ and
$(\mu \Psi, \iota_\Psi)$ for them, respectively.

\begin{definition}
The \emph{initial algebra correspondence for $\Phi$ and $\Psi$}
holds if
the following condition $(\ast)$ is satisfied.
\begin{align*}
  (\ast): \quad &\iota_\Phi \text{ is initial in }\AlgRst{\Phi}{\fixunit{\mathbb{C}}},\\
  &\text{and }\AlgFn{\id}{\rho}(\iota_\Psi) \text{ is initial in } \AlgRst{L\Phi R}{\fixcounit{\mathbb{D}}}.
\end{align*}
\end{definition}
\begin{proposition}
The initial algebra correspondence for $\Phi$ and $\Psi$ holds if and only if
$\iota_\Phi$ and $\AlgFn{\id}{\rho}(\iota_\Psi)$ lie in $\AlgRst{\Phi}{\fixunit{\mathbb{C}}}$ and $\AlgRst{L\Phi R}{\fixcounit{\mathbb{D}}}$, respectively, and they correspond to each other by
the adjoints $\AlgFn{L}{\alpha}$ and $\AlgFn{R}{\beta}$ given in Cor.~\ref{cor:fix}.
\qed
\end{proposition}
If $\rho$ is an isomorphism $L\Phi R \cong \Psi$, then
the correspondence reduces to the initial algebra correspondence for $\Phi$ and $L \Phi R$, already established as $(\star)$ in \S\ref{sec:init_alg_corresp}.
In what follows,
we focus on the cases where $\rho$ is not an isomorphism
though some of the components may be.

We aim to provide a sufficient condition ensuring this initial algebra correspondence,
particularly in the case where
$L$ is full and faithful.
Under this condition,
Prop.~\ref{prop:adj_full_faith} guarantees
$\iota_\Phi$ is initial in $\AlgRst{\Phi}{\fixunit{\mathbb{C}}}$.
Thus, we obtain the following result.
\begin{lemma} \label{lem:gra_psi}
  Assume that $L$ is full and faithful.
  The initial algebra correspondence for $\Phi$ and $\Psi$ holds if
  $\mu \Psi \in \fixcounit{\mathbb{D}}$ and
  $\AlgFn{\id}{\rho}(\iota_\Psi)$ is initial in $\Alg{L \Phi R}$.
  \qed
\end{lemma}
To establish the initial algebra correspondence for $\Phi$ and $\Psi$ when $L$ is full and faithful,
we assume $\mu \Psi \in \fixcounit{\mathbb{D}}$ as a reachability condition
and
aim to ensure that  $\AlgFn{\id}{\rho}(\iota_\Psi)$ is initial in $\Alg{L \Phi R}$ by constructing an isomorphism between initial chains.
\begin{definition}
  We say
the \emph{\gra of $\iota_\Psi$} is satisfied if
$\iota_\Psi \in \AlgRst{\Psi}{\fixcounit{\mathbb{D}}}$.
Equivalently, $\mu \Psi \in \fixcounit{\mathbb{D}}$.
\end{definition}

\begin{definition}[initial chain~\cite{Adamek1974}] \label{def:init_chain}
  We write  $\Ord$ for the category of ordinal numbers and their usual ordering.
  Let
  $\mathbb{E}$ be a category with all colimits of chains.
  For a functor $\Xi\colon \mathbb{E} \to \mathbb{E}$,
  the \emph{initial chain $W_\Xi$ of $\Xi$} is the functor $W_\Xi \colon \Ord \to \mathbb{E}$ defined
  (uniquely up to isomorphism) as follows:
  \begin{itemize}
    \item
  $W_\Xi(0) \coloneqq 0$ where $0$ is an initial object in $\mathbb{E}$,
  \item $W_\Xi(i+1) \coloneqq \Xi(W_\Xi(i))$ for each ordinal $i$. 
  \item
  $W_\Xi(i) \coloneqq \colim_{j < i}W_\Xi(j)$ for each limit ordinal $i$.
  \end{itemize}
  Morphisms are given by
  the unique map $W_\Xi(0, 1)$,
  and for successor steps $W_\Xi(j+1, i+1) = \Xi(W_\Xi(j, i))$.
  For limit ordinals $i$,
  the morphisms $W_\Xi(j, i) \ (j < i)$ form the colimiting cocone.
  For ordinals $i$,
  we write $\Xi^i 0$ for the object $W_\Xi(i)$.
  We say that the initial chain \emph{converges in $\lambda$ steps}
  if $W_\Xi(\lambda, \lambda+1)$ is an isomorphism.
\end{definition}

\begin{proposition}[\cite{Adamek1974}]
  If the initial chain converges in $\lambda$ steps, then $W_\Xi(\lambda, \lambda+1)^{-1}$ is an initial algebra.
  \qed
\end{proposition}

We extend the construction of initial chains to a functor.
It enables us to translate a natural transformation between functors to a natural transformation between their initial chains.
For omitted proofs, see 
\ifarxiv
Appendix~\ref{ap:omitted_L}.
\else
\cite[Appendix E]{arxiv}.
\fi
\begin{proposition} \label{prop:init_alg_isom}
In the setting of Def.~\ref{def:init_chain},
we define a functor $W_{\_}\colon [\mathbb{E}, \mathbb{E}] \to [\Ord, \mathbb{E}]$ mapping
$\Xi$ to $W_\Xi$ defined in Def.~\ref{def:init_chain} and $\rho\colon \Xi' \Rightarrow \Xi$ to $W_\rho \colon W_{\Xi'} \Rightarrow W_\Xi$ defined by
\begin{align*}
  (W_\rho)_0 &\coloneqq \id_0, \\
  (W_\rho)_{i+1} &\coloneqq \rho_{\Xi^i 0}\circ \Xi' ((W_\rho)_i) \text{ for each ordinal $i$}, \\
  (W_\rho)_i &\coloneqq \colim_{j < i} (W_\rho)_j \text{ for each limit ordinal $i$}.
\end{align*}
Then for a natural transformation $\rho\colon \Xi' \Rightarrow \Xi$,
$W_\rho$ is a natural isomorphism iff $\rho_{\Xi^i 0}$ is an isomorphism for each ordinal $i$.
\qed
\end{proposition}

\begin{lemma} \label{lem:init_alg_corresp_L}
  Assume that $\mathbb{D}$ has all colimits of chains.
  Then
  $\AlgFn{\id}{\rho}(\iota_\Psi)$ is initial in $\Alg{L\Phi R}$
 if the following conditions are satisfied:
  \begin{enumerate}
    \item \label{item:rho_isom} $\rho_{\Psi^i 0}$ is an isomorphism for each ordinal $i$.
    \item \label{item:psi_conv} The initial chain of $\Psi$ converges.
    \qed
  \end{enumerate}
\end{lemma}

Finally, the proposition below provides a sufficient condition for establishing the initial algebra correspondence.
\begin{proposition}\label{prop:init_alg_corresp_L}
  Assume that $\mathbb{D}$ has all colimits of chains and
  $L$ is full and faithful.
  The initial algebra correspondence for $\Phi$ and $\Psi$ holds
  if the following conditions are satisfied:
  \begin{enumerate}
    \item \label{item:rho_bar} $\overline{\rho} \colon \Phi R \Rightarrow R\Psi$ (resp.~$\rho_L \circ L \Phi \eta \colon L \Phi \Rightarrow \Psi L$) is an isomorphism.
    \item \label{item:converge} The initial chain of $\Psi$ converges.
    \item \label{item:init_chain_fix} The initial chain of $\Psi$ lies in $\fixcounit{\mathbb{D}}$.
  \end{enumerate}
  Furthermore,
  if the subcategory $\fixcounit{\mathbb{D}}\hookrightarrow \mathbb{D}$ is downward closed---that is, $Y \in \fixcounit{\mathbb{D}}$ and $X \to Y$ in $\mathbb{D}$ imply $X \in \fixcounit{\mathbb{D}}$---then \ref{item:init_chain_fix}) is implied by the \gra of $\iota_\Psi$.
  \qed
\end{proposition}
The proof proceeds by showing the two conditions in Lem.~\ref{lem:init_alg_corresp_L} (see 
\ifarxiv
Appendix~\ref{ap:omitted_L}).
\else
\cite[Appendix E]{arxiv}).
\fi
Once these conditions are satisfied,
$\rho_{\mu \Psi}$ becomes an isomorphism.
Thus, under the three conditions in Prop.~\ref{prop:init_alg_corresp_L},
$\AlgFn{\id}{\rho}(\iota_\Psi)$ is the morphism $\iota_\Psi$ composed with the isomorphism $\rho_{\mu \Psi}$.

The following example demonstrates the initial algebra correspondence under the \gra of $\iota_\Psi$.

\newcommand{\Acc}{\mathrm{Acc}}
\newcommand{\Ninf}{\mathbb{N}^\infty}
\begin{example}[unambiguous FA]
  \label{eg:unambiguousFA}
  \emph{Unambiguous finite automata (UFA)} (e.g.~\cite{Colcombet12,Colcombet15}) are non-deterministic finite automata such that there is at most one accepting run for each word.
  UFA has been applied for temporal verification (e.g.~\cite{BaierK00023}), focusing on its efficiency.
  For instance, some fundamental problems on automata such as equivalence, containment, and universality are all solvable in polynomial time for UFAs~\cite{StearnsH85,Seidl90}, while these are PSPACE-complete for NFAs.

  Consider an NFA $\langle \delta, \Acc\rangle\colon S \to \mathcal{P}(S)^A \times \mathbf{2}$ with finite set $S$ of states and finite set $A$ of alphabets.
  We instantiate our framework for
  two semantics of the NFA $c$: One is the recognized language
  and the other returns, for each word, its number of accepting paths.

  We define an endofunctor $\Phi$ on $[S \times A^*, \mathbf{2}]$ by
    \begin{displaymath}
      \Phi(k)(s, \vec{a}) \coloneqq
      \begin{cases}
        \Acc(s) &\text{if }\vec{a} = \epsilon, \\
        \bigvee_{s' \in \delta(s)(a)}k(s', \vec{a'}) &\text{if  }\vec{a} = a\vec{a'}.
      \end{cases}
    \end{displaymath}
  The domain $\mu \Phi$ of the initial algebra of $\Phi$
  characterizes the recognized language of the NFA.
  Specifically, $\mu \Phi(s, \vec{a}) = \top$ if the word $\vec{a}$ is accepted from the state $s$,
  and $\bot$ otherwise.

  Let $\Ninf$
be the complete lattice of extended natural numbers with the standard order,
and let $i\colon \mathbf{2} \to \Ninf$ be the function sending $\top$ to $1$ and $\bot$ to $0$.
  We define an endofunctor $\Psi$ on $[S \times A^*, \Ninf]$ by
  \begin{displaymath}
    \Psi(k)(s, \vec{a}) \coloneqq \begin{cases}
   i\big(\Acc(s)\big) &\text{if } \vec{a} = \epsilon, \\
    \Sigma_{s' \in \delta(s)(a)}k(s', \vec{a'}) &\text{if } \vec{a} = a\vec{a'}.
    \end{cases}
  \end{displaymath}
  The domain $\mu \Psi$ of the initial algebra of $\Psi$ assigns to $(s, \vec{a})$ the number of accepting runs of $\vec{a}$ from $s$.

    We show that under a \gra{} representing the unambiguity of NFAs, the initial algebras of $\Phi$ and $\Psi$ coincide.
    Consider
    the adjunction
    \begin{displaymath}
      L\dashv R\colon [S \times A^*, \Ninf] \to [S \times A^*, \mathbf{2}]
    \end{displaymath}
    defined by
    (i) $L(k)(s, \vec{a}) \coloneqq i\big(k(s, \vec{a})\big)$;
    (ii) $R(k)(s, \vec{a}) \coloneqq \top$ if $1 \leq k(s, \vec{a})$, and $\bot$ otherwise.
    Since
    $(L \Phi  R)(k)(s, \vec{\alpha}) = \min(1, \Psi(k)(s, \vec{\alpha}))$
    for each $k\in [S \times A^*, \Ninf]$ and $(s, \vec{\alpha}) \in S \times A^*$,
    there is a unique natural transformation $\rho\colon L \Phi R \Rightarrow \Psi$.

    All three conditions of Prop.~\ref{prop:init_alg_corresp_L} are satisfied
    under the \gra of $\iota_\Psi$ (since $[S \times A^*, \Ninf]$ is downward closed; see 
    \ifarxiv
    Prop.~\ref{ap:prop:unambiguous}).
    \else
    \cite[Prop.~E.1]{arxiv}).
    \fi
Consequently, under the \gra of $\iota_\Psi$, we obtain the initial algebra correspondence of $\Phi$ and $\Psi$.
In this setting,
the \gra
$\mu \Psi \in \fixcounit{\mathbb{D}}$ means that
for any state and any word,
there is at most one accepting path.
This condition precisely captures the unambiguity of the NFA.
\end{example}

Instantiated to UFAs, our framework shows that two semantics coincide: the first is the original problem we want to solve (i.e.~computing accepting words), and the second semantics is the equivalent problem under the unambiguity condition (i.e.~counting the number of accepting runs for each word). Using this correspondence, we can rely on existing efficient computation methods (e.g.~\cite{DBLP:conf/dcfs/Colcombet15}) for the universality of UFAs based on the second semantics.

  For simplicity,
in the previous example
we focus on how the unambiguity condition of an NFA establishes a correspondence that also encodes its unambiguity.
We can extend this example to illustrate
a correspondence between two different methods of computing probabilities of words, again
under the unambiguity of NFAs.
See 
\ifarxiv
Appendix~\ref{ap:unambiguous_mc} 
\else
\cite[Appendix F]{arxiv}
\fi
for details.

\section{Predicate transformers via Kan extensions}
\label{sec:kanExtension}
We have focused on initial algebra correspondences
arising from the adjoint equivalences in Cor.~\ref{cor:fix}.
We now turn to how we can derive target predicate transformers for these initial algebra correspondences.
In particular, we consider
predicate transformers of the forms $\Phi$ and $L\Phi R$ (as in the framework of Sec.~\ref{sec:init_alg_corresp}).
In that setting,
once $\Phi$ is specified, the composite $L\Phi R$ arises naturally as the best approximation of $\Phi$ along the adjunction $L \dashv R$.
However, in certain cases such as those described in Example~\ref{eg:achievable_mdp},
it is not immediately clear how to define the functor $\Phi$ itself from the given problem setting.

To handle such cases,
we propose a method of constructing predicate transformers via pointwise Kan extensions.
This pointwise approach naturally yields
a predicate transformer of MDPs that aggregates results across all non-deterministic choices.
Moreover,
we identify a sufficient condition for two predicate transformers defined in this way
to fit into
the form $\Phi$ and $L\Phi R$,
allowing us to apply Cor.~\ref{cor:fix}.

Throughout this section,
all categories are assumed to be locally small.

\subsection{Review of Kan Extensions}
We briefly review concepts of (pointwise) Kan extensions,
which provide a way of extending functors in a universal manner.
We refer the reader to~\cite{mac2013categories, Loregian_2021} for further details.
\newcommand{\leftKan}[2]{\mathrm{Lan}_{#1}{#2}}
\begin{definition}[Left Kan extensions]
  Let $K\colon \mathbb{A} \to \mathbb{C}$ and $T\colon \mathbb{A} \to \mathbb{B}$ be functors.
  A \emph{left Kan extension of $T$ along $K$} is a functor $\leftKan{K}{T}\colon \mathbb{C} \to \mathbb{B}$ equipped with a natural transformation $\alpha\colon T \Rightarrow \leftKan{K}{T} \circ K$
  such that
  for each pair $(S\colon \mathbb{C} \to \mathbb{B}, \beta\colon T \Rightarrow S \circ K)$,
  there exists a unique natural transformation $\gamma\colon \leftKan{K}{T} \Rightarrow S$ satisfying $\gamma K \circ \alpha = \beta$.
\end{definition}

\begin{proposition} \label{prop:kan_coend}
  Let $K\colon \mathbb{A} \to \mathbb{C}$ and $T\colon \mathbb{A} \to \mathbb{B}$ be functors,
  and assume that
  $\mathbb{A}$ is small and
  $\mathbb{B}$ is cocomplete.
  Since $\mathbb{B}$ is locally small and it has coproducts,
  it is canonically copowered over $\Set$
  (i.e.~there exists a functor $(-) \bullet (-)\colon \Set \times \mathbb{B} \to \mathbb{B}$ with an isomorphism
  $\mathbb{B}(X \bullet y, z) \cong \Set(X, \mathbb{B}(y, z))$ natural in all components)
  by $X \bullet y \coloneqq \coprod_{x \in X} y$.
  Since $\mathbb{A}$ is small and $\mathbb{B}$ is cocomplete,
  the coend below exists.
  Then a left Kan extension of $T$ along $K$ exists,
  and there is an isomorphism
  \begin{align}
    \leftKan{K}{T} &\cong \int^{a} \mathbb{C}(Ka, -) \bullet T(a),
  \end{align}
  natural in $K$ and $T$.
  Moreover, the left Kan extension is \emph{pointwise}, meaning that for each $c \in \mathbb{C}$,
    $\leftKan{K}{T}(c) \cong \mathrm{colim} (K \downarrow c \xrightarrow{P} \mathbb{A} \xrightarrow{T} \mathbb{B})$
    where $P$ is the projection functor of the slice category $(K \downarrow c)$.
    \qed
\end{proposition}

\begin{proposition} \label{prop:kan_adj}
  Left adjoints preserve pointwise left Kan extensions.
  Concretely,
  if
  $L\colon \mathbb{B} \to \mathbb{D}$ is a left adjoint and
  $T\colon \mathbb{A} \to \mathbb{B}, K\colon \mathbb{A} \to \mathbb{C}$
  are functors such that
  a left Kan extension $(\leftKan K T, \alpha)$ exists,
  then
  $(L \circ \leftKan K T, L \alpha)$ is a left Kan extension of $L \circ T$ along $K$.
  \qed
\end{proposition}

\subsection{Constructing Predicate Transformers with Kan Extensions}
We now introduce a method for extending
$\Phi\colon \mathbb{I} \times \mathbb{A} \to \mathbb{A}$, where $\mathbb{I}$ is a category representing \emph{parameters}.
We refer to such
a functor $\Phi$ as a \emph{parameterized endofunctor};
fixing a parameter $i \in \mathbb{I}$ yields a predicate transformer $\Phi(i, -)\colon \mathbb{A} \to \mathbb{A}$.
In the definition below,
we derive a new parameterized endofunctor from
$\Phi$ by changing the parameter category $\mathbb{I}$ as well as the domain and codomain categories $\mathbb{A}$.

\newcommand{\PhiKF}[3]{#1_{#2, #3}}
\begin{definition} \label{def:pred_kan}
  Let $\mathbb{I}$ and $\mathbb{A}$ be small categories,
  $\mathbb{J}$ be a category,
  $\mathbb{B}$ be a cocomplete category.
  Suppose
  $K\colon \mathbb{I} \to \mathbb{J}$,
  $T\colon \mathbb{A} \to \mathbb{B}$
  and
  $\Phi\colon \mathbb{I} \times \mathbb{A} \to \mathbb{A}$.
  A \emph{$(K, T)$-extension $\PhiKF{\Phi}{K}{T}$ of $\Phi$}
  is a left Kan extension
  of $T \circ \Phi$ along $K \times T$.
  \begin{center}
  \begin{tikzcd}
      \mathbb{J} \times \mathbb{B} \arrow[drr, "\PhiKF{\Phi}{K}{T} = \leftKan{K \times T}{T \circ \Phi}", ""{name=PhiKF, below, near start}] \\
      \mathbb{I} \times \mathbb{A} \arrow[u, "K \times T"] \arrow[r, "\Phi"'{name=Phi}]
      \arrow[Rightarrow, u, to=PhiKF.south, shorten >=2pt]
      &\mathbb{A} \arrow[r, "T"'] &\mathbb{B}
  \end{tikzcd}
  \end{center}

  Since $\mathbb{I} \times \mathbb{A}$ is small and $\mathbb{B}$ is cocomplete,
  such a left Kan extension always exists by
Prop.~\ref{prop:kan_coend}. Furthermore, for each $j \in \mathbb{J}$ and $b \in \mathbb{B}$, the following isomorphism holds:
\begin{align}
  \PhiKF{\Phi}{K}{T}(j, b)
  &\cong \mathrm{colim} ((K \times T) \downarrow (j, b) \to \mathbb{I} \times \mathbb{A} \xrightarrow{T \circ \Phi} \mathbb{B}).\label{eq:pred_kan_colim}
\end{align}
\end{definition}

By applying
this method for extending predicate transformers,
we can show that the Bellman operators for MDPs, denoted by $\Phi$ and $L \Phi R$ in Example~\ref{eg:achievable_mdp}, arise systematically from
the Bellman operator for MCs, $\Phi$ in Example~\ref{eg:total_partial_mc}.
Concretely,
these Bellman operators for MDPs
can be defined as an aggregation of the results
of the Bellman operator for MCs over all non-deterministic choices,
and
this aggregation can be captured
by
the colimit in \eqref{eq:pred_kan_colim}.

Before discussing Bellman operators in detail,
we review some concepts and notations related to posets and downward closures in a categorical context.
We denote by
$\Poset$ and $\SupLat$ the categories of posets (and order-preserving maps)
and suplattices (and join-preserving maps), that are posets with arbitrary joins, respectively.
These two categories together with $\Set$ are related by the following adjunctions:
\begin{displaymath}
  \xymatrix@C=2em{
    \Set \ar@/^0.8em/[rr]^-{\disc} &\bot
&\Poset \ar@/^0.8em/[ll]^-U \ar@/^0.8em/[rr]^-{(-)^\downarrow} &\bot &\SupLat. \ar@/^0.8em/[ll]^-U
  }
\end{displaymath}
The functors $U$ represent forgetful functors,
the functor $\disc$ sends a set $X$ to the poset $X$ with the discrete order,
and
$(-)^\downarrow$ sends $(X, \sqsubseteq)$ to the complete lattice of downward-closed subsets $(X^\downarrow, \subseteq)$.
In addition, we have $\disc \circ U = \id$.

Let $\Omega \in \SupLat$ be the complete lattice $\interval{1} \times \interval{\infty}$.
In what follows,
we basically work within $\Poset$.
We often
omit writing $U$ and $\disc$ explicitly;
for example, we write simply $S$ for $\disc(S) \in \Poset$
and
$\Omega^\downarrow$ for $U((U\Omega)^\downarrow) \in \Poset$.

The category $\Poset$ is cartesian closed
with exponential objects $[X, Y]$.
Given a morphism $f\colon X \to Y$,
we write $f_*$ for the induced map $[S, f] = f \circ (-)\colon [S, X] \to [S, Y]$.

\newcommand{\Phimc}{\Phi^{\mathrm{mc}}}
\begin{example}[Bellman operator for partial expected reward] \label{eg:vi_mc_mdp_partial}
  Recall that
  in Example~\ref{eg:total_partial_mc},
  for a given MC $c\colon S \to FS$
  where $F = \mcfunc{-} \times \mathbb{N}$,
  we introduced a functor $\Phi\colon [S, \Omega] \to [S, \Omega]$ where $\Omega = \interval{1} \times \interval{\infty}$.
This $\Phi$ can be extended to a functor of the type
  $[S, FS] \times [S, \Omega] \to [S, \Omega]$,
  and we denote it by $\Phimc$.
  The functor $\Phimc$ represents the standard Bellman operator calculating both reachability probabilities and partial expected rewards for MCs.

  Let us obtain the Bellman operator
  for the partial expected reward for MDPs, denoted by $\Phi$ in Example~\ref{eg:achievable_mdp},
  as an extension of $\Phimc$.
  In particular, by Def.~\ref{def:pred_kan},
  we obtain
\begin{displaymath}
  \PhiKF{\Phimc}{(\eta_{FS})_*}{(\eta'_{\Omega})_*}\colon [S, \Pfin FS] \times [S, \Omega^\downarrow] \to [S, \Omega^\downarrow]
\end{displaymath}
  where
  $\eta$ is the unit of the monad $\Pfin$
  and
  $\eta'$ is the unit of $(-)^\downarrow \dashv U$.
  Note that it is unique
  since $[S, \Omega^\downarrow]$ is a poset.
  By the colimit expression~\eqref{eq:pred_kan_colim} in this context,
  for $c \in [S, \Pfin FS]$ and $k \in [S, \Omega^\downarrow]$, we have
  \begin{align*}
  &\PhiKF{\Phimc}{(\eta_{FS})_*}{(\eta'_{\Omega})_*}(c, k) \\
  &= \left(s \mapsto \bigcup_{(c', k') \in ((\eta_{FS})_* \times (\eta'_{\Omega})_*) \downarrow (c, k)} (\Phimc(c', k')(s))^\downarrow \right)
  \end{align*}
  This coincides with $\Phi(k)$ in Example~\ref{eg:achievable_mdp}.
  In more detail,
  the slice category $((\eta_{FS})_* \times (\eta'_{\Omega})_*) \downarrow (c, k)$
  consists of
  pairs $(c', k')$ of $c'\colon S \to FS$ and $k'\colon S \to \Omega$
  such that $c'(s') \in c(s')$ and $k'(s') \in k(s')$ for each $s \in S$.
  A morphism
  $(c', k') \to (c'', k'')$ exists if and only if $c' = c''$ and $k'(s') \leq k''(s')$ for each $s' \in S$.
  Intuitively,
  selecting $c'$ in each step can be viewed as selecting a particular choice on the MDP $c$.
  By taking
  the union over all such $c'$ and $k'$,
  we aggregate
  the results of $\Phimc$ for all optimal choices.
  In this way,
  $\PhiKF{\Phimc}{(\eta_{FS})_*}{(\eta'_{\Omega})_*}$
  encodes a Bellman operator for MDPs naturally derived from $\Phimc$.
\end{example}

\begin{example}[Bellman operator for total expected reward] \label{eg:vi_mc_mdp_total}
  In a similar way to Example~\ref{eg:vi_mc_mdp_partial}, one obtains
  \begin{displaymath}
  \PhiKF{\Phimc}{(\eta_{FS})_*}{(\pi_2)_*}\colon [S, \Pfin FS] \times [S, \interval{\infty}] \to [S, \interval{\infty}]
  \end{displaymath}
  where
  $\eta$ is the unit of the monad $\Pfin$.
  For each $c \in [S, \Pfin FS]$ and $k \in [S, \interval{\infty}]$,
  \begin{align*}
  &\PhiKF{\Phimc}{(\eta_{FS})_*}{(\pi_2)_*}(c, k) \\
  &= \left(s \mapsto \bigsqcup_{(c', k') \in ((\eta_{FS})_* \times (\pi_2)_*) \downarrow (c, k)} \pi_2 \big(\Phimc(c', k')(s)\big)\right),
  \end{align*}
  and it represents
  the Bellman operator for the total expected reward, denoted by $L \Phi R$
  in Example~\ref{eg:achievable_mdp}.
  The functor $\PhiKF{\Phimc}{(\eta_{FS})_*}{(\pi_2)_*}$
  collects the second components of the results of $\Phimc$ over all choices, merging them by joins.
\end{example}

\subsection{Relation Between Extensions of Predicate Transformers}
As illustrated in the previous subsection,
our construction of $\PhiKF{\Phi}{(-)}{(-)}$ in Def.~\ref{def:pred_kan}
provides a systematic way to derive parameterized endofunctors from a given parameterized endofunctor $\Phi$
by aggregating results for decompositions of inputs (coalgebras and predicates).

The next proposition
establishes conditions under which
two extensions of a parameterized endofunctor
fit the framework in Cor.~\ref{cor:fix}.

\begin{proposition} \label{prop:kan_isom}
  Let $\Phi\colon \mathbb{I} \times \mathbb{A} \to \mathbb{A}$.
  Consider two extensions
  $\PhiKF{\Phi}{K_1}{T_1}$ and $\PhiKF{\Phi}{K_2}{T_2}$
  where $K_i\colon \mathbb{I} \to \mathbb{J}_i, T_i\colon \mathbb{A} \to \mathbb{B}_i$ $(i=1, 2)$.
  If we have a full and faithful functor $H\colon \mathbb{J}_2 \to \mathbb{J}_1$ and an adjunction $L \dashv R\colon \mathbb{B}_2 \to \mathbb{B}_1$ such that $L \circ T_1 = T_2$ and $K_1 = H \circ K_2$,
  then $L \PhiKF{\Phi}{K_1}{T_1} (H \times R) \cong \PhiKF{\Phi}{K_2}{T_2}$.
  \begin{displaymath}
    \xymatrix{
      \mathbb{I} \ar[r]^{K_1} \ar[dr]_{K_2} &\mathbb{J}_1  \\
      &\mathbb{J}_2 \ar[u]_{H}
    },
    \xymatrix{
      \mathbb{A} \ar[r]^{T_1} \ar[dr]_{T_2} &\mathbb{B}_1 \ar[d]_L^\dashv \\
      &\mathbb{B}_2 \ar@/_3ex/[u]_R
    }
  \end{displaymath}
\end{proposition}
\begin{proof}
  By Prop.~\ref{prop:kan_adj} and $L \dashv R$,
\begin{align*}
  &L \PhiKF{\Phi}{K_1}{T_1} (H \times R) \\
   &\cong \int^{i, a} (\mathbb{J}_1(K_1(i), H-) \times \mathbb{B}_1(T_1(a), R-)) \bullet L T_1 \Phi(i, a) \\
   &\qquad \qquad \qquad \qquad \qquad \qquad \text{(since $H$ is full and faith)} \\
   &\cong \int^{i, a} (\mathbb{J}_1(K_2(i), -) \times \mathbb{B}_1(LT_1(a), -)) \bullet L T_1 \Phi(i, a) \\
   &= \PhiKF{\Phi}{K_2}{T_2}
\end{align*}
\end{proof}

By applying Cor.~\ref{cor:fix}
to this extension $\PhiKF{\Phi}{K_1}{T_1}$ and $L \dashv R$,
we obtain
an equivalence of subcategories of algebras.
\begin{corollary} \label{cor:kan_equiv}
  In the setting of Prop.~\ref{prop:kan_isom},
  for each $j \in \mathbb{J}_1$,
  \begin{displaymath}
  \AlgRst{\PhiKF{\Phi}{K_1}{T_1}(H(j), -)}{\fixunit{(\mathbb{B}_1)}}
  \simeq \AlgRst{\PhiKF{\Phi}{K_2}{T_2}(j, -)}{\fixcounit{(\mathbb{B}_2)}}
  \end{displaymath}
  by functors induced by the adjunction $L \dashv R$.
  \qed
\end{corollary}

We close this section with two examples illustrating relationships among Bellman operators discussed so far.
\begin{example}
  In Example~\ref{eg:vi_mc_mdp_partial},
  we introduced the Bellman operator $\Phimc$ for MCs
  and
  the operator $\PhiKF{\Phimc}{(\eta_{FS})_*}{(\eta'_\Omega)_*}$ for the partial expected reward for MDPs as an extension of $\Phimc$.
  Noting that $\Phimc \cong \PhiKF{\Phimc}{\id}{\id}$,
  these two Bellman operators can be seen as two extensions of $\Phimc$.

  We have the following diagrams.
  Since $\Omega \in \SupLat$, there is a left adjoint $\sqcup\colon \Omega^\downarrow \to \Omega$ of $\eta'_\Omega$.
  \begin{displaymath}
    \xymatrix{
      [S, FS] \ar[r]^{(\eta_{FS})_*} \ar[dr]_{\id} &[S, \Pfin FS]  \\
      &[S, FS] \ar[u]_{(\eta_{FS})_*}
    },
    \xymatrix{
      [S, \Omega] \ar[r]^{(\eta'_\Omega)_*} \ar[dr]_{\mathrm{id}} &[S, \Omega^\downarrow] \ar[d]_{\sqcup_*}^\dashv \\
      &[S, \Omega] \ar@/_3ex/[u]_{(\eta'_\Omega)_*}
    }
  \end{displaymath}
  Since $(\eta_{FS})_*$ is full and faithful,
Prop.~\ref{prop:kan_isom} induces
\begin{displaymath}
  \sqcup_* \PhiKF{\Phimc}{(\eta_{FS})_*}{(\eta'_\Omega)_*} ((\eta_{FS})_* \times (\eta'_\Omega)_*) \cong \PhiKF{\Phimc}{\id}{\id} \cong \Phimc.
\end{displaymath}

Hence, by Cor.~\ref{cor:kan_equiv},
for each MC $c\colon S \to FS$,
there is
an equivalence between
  $\AlgRst{\PhiKF{\Phimc}{(\eta_{FS})_*}{(\eta'_\Omega)_*}(\eta_{FS} \circ c, -)}{\fixunit{[S, \Omega^\downarrow]}}$
  and $\Alg{\Phimc(c, -)}$ induced by $\sqcup_* \dashv (\eta'_\Omega)_*$.

The functor $\PhiKF{\Phi}{(\eta_{FS})_*}{(\eta'_\Omega)_*}(\eta_{FS} \circ c, -)\colon [S, \Omega^\downarrow] \to [S, \Omega^\downarrow]$ coincides with the function $\Phi$ in Example~\ref{eg:achievable_mdp} for the MDP $\eta_{FS} \circ c$.
In this setting,
$f$ and $\mu \Phi$ in Example~\ref{eg:achievable_mdp} are equal.
Moreover,
under the adjunction $\sqcup_* \dashv (\eta'_\Omega)_*$,
the \gra of $\iota_\Phi$ is always satisfied.
Consequently, Prop.~\ref{prop:gra_correspondence} provides
the initial algebra correspondence for
the Bellman operator for the MC $c$
and that for the MDP $\eta_{FS} \circ c$.
\end{example}

\begin{example}
  In Example~\ref{eg:vi_mc_mdp_partial} and~\ref{eg:vi_mc_mdp_total},
  we introduced two Bellman operators $\PhiKF{\Phimc}{(\eta_{FS})_*}{(\eta'_\Omega)_*}$
  and $\PhiKF{\Phimc}{(\eta_{FS})_*}{(\pi_2)_*}$,
  which compute
  partial and total expected rewards, respectively.
  The following diagrams show a relation between these two operators.
  \begin{displaymath}
    \xymatrix@C=3em{
      [S, FS] \ar[r]^{(\eta_{FS})_*} \ar[dr]|-{(\eta_{FS})_*} \ar@/_.5em/[ddr]_{(\eta_{FS})_*} &[S, \Pfin FS]  \\
      &[S, \Pfin FS] \ar[u]_{\id} \\
      &[S, \Pfin FS] \ar[u]_{\id}
    },
    \xymatrix@C=3em{
      [S, \Omega] \ar[r]^{(\eta'_\Omega)_*} \ar[dr]|{\id} \ar@/_.5em/[ddr]_{(\pi_2)_*} &[S, \Omega^\downarrow] \ar[d]_{\sqcup_*}^\dashv \\
      &[S, \Omega] \ar@/_3ex/[u]_{(\eta'_\Omega)_*} \ar[d]_-(.4){(\pi_2)_*\!\!}^\dashv\\
      &[S, \interval{\infty}] \ar@/_3ex/[u]_{\langle \Delta_1, \id\rangle_*}
    }
  \end{displaymath}
  Then Prop.~\ref{prop:kan_isom} yields the following isomorphism.
  \begin{displaymath}
    (\pi_2  \circ \sqcup)_* \PhiKF{\Phimc}{(\eta_{FS})_*}{(\eta'_\Omega)_*} (\id \times (\eta'_\Omega \circ \langle \Delta_1, \id\rangle)_*) \cong \PhiKF{\Phimc}{(\eta_{FS})_*}{(\pi_2)_*}.
  \end{displaymath}

  By Cor.~\ref{cor:kan_equiv},
  for each MDP $c\colon S \to \Pfin FS$,
  we have
  an equivalence between
  $\AlgRst{\PhiKF{\Phimc}{(\eta_{FS})_*}{(\eta_\Omega)_*}(c, -)}{\fixunit{[S, \Omega^\downarrow]}}$
  and $\Alg{\PhiKF{\Phimc}{(\eta_{FS})_*}{(\pi_2)_*}(c, -)}$ induced by the adjunction $(\pi_2 \circ \sqcup)_* \dashv (\eta'_\Omega \circ \langle \Delta_1, \id\rangle)_*$.
  This equivalence captures the correspondence between an $f$ and $\mu (L \Phi R)$ described in Example~\ref{eg:achievable_mdp}
  for the Bellman operator for the partial expected reward and that for the total expected reward,
under the \gra.
\end{example}

\section{Conclusion}
\label{sec:conclusion}
We investigated reachability conditions that ensure the initial algebra correspondence
between two different semantics
connected by adjunctions.
Our framework accommodates several examples, including the almost sure reachability condition for MCs and the unambiguity condition for NFAs.

There is a coalgebraic construction of predicate transformers with modalities by fibrational liftings (e.g.~\cite{DBLP:journals/mscs/HasuoKC18}).
We plan to study abstract reachability conditions with modalities, which leads to the initial algebra correspondence of the predicate transformers that are induced by modalities.
Another direction is to study the compositionality of \gra{}s, that is, to study which compositions of systems, for instance, parallel composition, preserve a given \gra{}.
Perhaps the abstract GSOS framework~\cite{TuriP97} is helpful to achieve this goal.
It is also interesting to see whether our abstract framework can be extended to accommodate initial algebra correspondences under linear time properties beyond reachability conditions.
  Additionally,
we believe that our results on unambiguous FA in Example~\ref{eg:unambiguousFA} could form a basis for recovering the PTIME algorithm for model checking of MCs and UFAs~\cite{BaierK00023} from a categorical framework. Specifically, 
integrating our framework with the unifying theory for coalgebraic product construction~\cite{DBLP:journals/pacmpl/WatanabeJRH25}
and utilizing the categorical framework to ensure uniqueness of fixed points presented in~\cite{DBLP:conf/concur/KoriHK21}
may facilitate this recovery.

\section*{Acknowledgments}
The authors would like to thank the anonymous reviewers for their valuable comments and suggestions.
We are supported by the ASPIRE grant No.~JPMJAP2301, JST. 
KW was supported by the JST grant No.~JPMJAX23CU.

\bibliographystyle{IEEEtran}
\bibliography{mybib}

\ifarxiv

\appendices
\input{appendix.tex}
\else
\fi

\end{document}

%% file: appendix.tex
\clearpage 

\section{Omitted Proofs for \S\ref{sec:cat}} \label{ap:omitted_cat}
\begin{proof}[Proof of Thm.~\ref{thm:hermida_sub_equiv}]
  We first show that
  \begin{displaymath}
  \adj {\AlgFn{L}{\alpha}}{\AlgFn{R}{\beta}}{\eta}{\epsilon}\colon \AlgRst{\Phi}{\mathbb{C}'} \to \AlgRst{\Psi}{\mathbb{D}'}
  \end{displaymath}
forms an adjunction.
  A morphism $f\colon (X, a) \to \AlgFn{R}{\beta}(Y, b)$
  is $f\colon X \to R(Y)$ in $\mathbb{C}'$ such that $f \circ a = R(b) \circ \beta_Y \circ \Phi (f) =  \overline{b \circ \Psi\epsilon_Y \circ \alpha^{-1}_{R(Y)}} \circ \Phi (f)$.
  Its transpose is $\overline{f}\colon L(X) \to Y$ in $\mathbb{D}'$ such that
  $\overline{f} \circ L(a) = b \circ \Psi\epsilon_Y \circ \alpha^{-1}_{R(Y)} \circ L\Phi (f)$.
  By naturality of $\alpha$,
  it follows that $b \circ \Psi\epsilon_Y \circ \alpha^{-1}_{R(Y)} \circ L\Phi (f)
  = b \circ \Psi\epsilon_Y \circ \Psi L(f) \circ \alpha^{-1}_X
  = b \circ \Psi (\overline{f}) \circ \alpha^{-1}_X$.
  Therefore
  such $\overline{f}$ coincides with
  $\overline{f}\colon \AlgFn{L}{\alpha}(X, a) \to (Y, b)$.

  Next, assume $\mathbb{C}' = \fixunit{\mathbb{C}}$ and $\mathbb{D}' = \fixcounit{\mathbb{D}}$ with the adjunction $\adj L R {\eta}{\epsilon}\colon \mathbb{C} \to \mathbb{D}$.
  Then
  $\AlgFn{L}{\alpha}$ and $\AlgFn{R}{\beta}$ form an equivalence of categories
  because
  $\eta\colon \id \cong \AlgFn{R}{\beta}\AlgFn{L}{\alpha}$
  and $\epsilon\colon \AlgFn{L}{\alpha}\AlgFn{R}{\beta} \cong \id$.
\end{proof}

\section{Definition of Total and Partial Expected Rewards for MDPs}
\label{sec:defScheduler}

We fix $F \coloneqq \mathcal{D}(- + \{\checkmark\}) \times \mathbb{N}$ in this section.
Total and partial expected rewards for MCs are also defined naturally by the following definitions
since MCs are special cases of MDPs.
\begin{definition}[scheduler]
  Let $c\colon S\rightarrow  \Pfin FS$ be an MDP with countable states (defined in Example~\ref{eg:achievable_mdp}).
  A \emph{(deterministic) scheduler} is a function $\sigma\colon S^{+}\rightarrow FS$ s.t. for any sequence $s_1\cdots s_n$, $\sigma(s_1\cdots s_n)\in c(s_n)$.
\end{definition}

A pair of an initial state and a scheduler on a given MDP induces an MC with countably infinite state space, whose states are reachable states with histories from the initial state under the scheduler.
The reachability probability objective and expected reward objective under its scheduler on the MDP are defined on the MC with countable states.
See~\cite{DBLP:books/daglib/0020348,ForejtKNP11} for further details.

\newcommand{\Path}{\mathrm{Path}}
\begin{definition}[Total and partial expected reward] \label{def:p_erew_n}
  Let $\sigma$ be a scheduler. For each $n \in \mathbb{N}$,
  define
  $\Path^n(s) \coloneqq
  \{(s_1\cdots s_m) \in S^m \mid m \leq n, s_1 = s \}$.
  For each $(s_1\cdots s_m) \in \Path^n(s)$,
  let
  \begin{align*}
  \Prb{\sigma}((s_1\cdots s_m)) &\coloneqq \prod_{i=1}^{m-1} \pi_1 (\sigma(s_1\cdots s_i))(s_{i+1}), \\
  \rew_\sigma((s_1\cdots s_m)) &\coloneqq \Sigma_{i=1}^{m} \pi_2 (\sigma(s_1\cdots s_i)).
  \end{align*}
  By abuse of notation, we
  write
  $\Prb{\sigma}((s_1\cdots s_m, \checkmark))$ for $\prod_{i=1}^m \pi_1 (\sigma(s_1\cdots s_i))(s_{i+1})$
  where $s_{m+1} \coloneqq \checkmark$.

  We then define
  \begin{align*}
  \Prb{\sigma}^n(s, \checkmark) &\coloneqq \Sigma_{\vec{s} \in \Path^n(s)} \Prb{\sigma}(\vec{s}, \checkmark), \\
  \Erew{\sigma}^n(s, \checkmark) &\coloneqq \Sigma_{\vec{s} \in \Path^n(s)} \Prb{\sigma}(\vec{s}, \checkmark) \cdot \rew_\sigma(\vec{s}), \\
  \Prb{\sigma}(s, \checkmark) &\coloneqq \sup_{n \in \mathbb{N}} \Prb{\sigma}^n(s, \checkmark), \\
  \Erew{\sigma}(s, \checkmark) &\coloneqq \sup_{n \in \mathbb{N}} \Erew{\sigma}^n(s, \checkmark).
  \end{align*}
  The $\Prb{\sigma}(s, \checkmark)$ and $\Erew{\sigma}(s, \checkmark)$
  is the probability and the \emph{partial expected reward} of reaching the target state $\checkmark$ from $s$ under the scheduler $\sigma$, respectively.

  For the total expected reward, we regard $\checkmark$ as a sink state such that there is a unique action that induces the self-loop with $0$ reward.
  We then define the set $\Path^{=n}(s)$ as  $
  \{(s_1\cdots s_n) \in (S+\{\checkmark\})^n \mid s_1 = s\}$.
  The \emph{total expected reward} from $s$ under the scheduler $\sigma$ is defined by
  \begin{displaymath}
    \sup_{n\in \mathbb{N}} \Sigma_{\vec{s} \in\Path^{=n}(s)} \Prb{\sigma}(\vec{s}) \cdot \rew_\sigma(\vec{s}).
  \end{displaymath}
\end{definition}

\section{Omitted proofs for \S\ref{sec:init_alg_corresp}} \label{ap:omitted_mdp}
\begin{lemma} \label{lem:f_init}
  Consider Example~\ref{eg:achievable_mdp}.
  We write $\Delta_{(0, 0)^\downarrow}$ for the constant function $S \to (\interval{1} \times \interval{\infty})^\downarrow$ of the value
  $(0, 0)^\downarrow$.
  \begin{enumerate}
    \item \label{item:f_alg} $f$ consists a $\Phi$-algebra.
    \item \label{item:phi_delta} For each $n \in \mathbb{N}$ and $s \in S$,
    \begin{displaymath}
      \Phi^n(\Delta_{(0, 0)^\downarrow})(s) = \bigcup_{\text{scheduler }\sigma}\big(\Prb{\sigma}^{n}(s, \checkmark), \Erew{\sigma}^{n}(s, \checkmark)\big)^\downarrow,
    \end{displaymath}
where $\Prb{\sigma}^{n}(s, \checkmark)$
and $\Erew{\sigma}^{n}(s, \checkmark)$ are defined in Def.~\ref{def:p_erew_n}.
    \item \label{item:cup_k} For each object $(k, \Phi k \sqsubseteq k) \in \AlgRst{\Phi}{\fixunit{[S, (\interval{1} \times \interval{\infty})^\downarrow]}}$,
    $\bigsqcup_{n \in \mathbb{N}}\Phi^n(\Delta_{(0, 0)^\downarrow}) \sqsubseteq k$.
    \item If $f \in \fixunit{[S, (\interval{1} \times \interval{\infty})^\downarrow]}$,
    then $(f, \Phi f \sqsubseteq f)$ is initial in $\AlgRst{\Phi}{\fixunit{[S, (\interval{1} \times \interval{\infty})^\downarrow]}}$.
  \end{enumerate}
\end{lemma}
\begin{proof}
  1) For each $s \in S$,
  \begin{align*}
    &\Phi(f)(s) \\
    &= \bigcup_{\substack{(d, e) \in c(s),
    k'\colon S \to \interval{1} \times \interval{\infty} \text{ s.t.~} \\\forall s'.~k'(s') \in \bigcup_{\sigma}(\Prb{\sigma}(s', \checkmark), \Erew{\sigma}(s', \checkmark))^\downarrow}}
    \Big(\big(p(d, e, k'), r(d, e, k')\big)\Big)^{\downarrow} \\
    &\subseteq \bigcup_{\sigma}(\Prb{\sigma}(s, \checkmark), \Erew{\sigma}(s, \checkmark))^\downarrow.
  \end{align*}
  The last inclusion holds because
  for each $(d, e) \in c(s)$ and $k'\colon S \to \interval{1} \times \interval{\infty}$  such that $\forall s'.~k'(s') \in \bigcup_{\sigma}(\Prb{\sigma}(s', \checkmark), \Erew{\sigma}(s', \checkmark))^\downarrow$,
  there is a scheduler
  $\sigma$ defined by
  $\sigma(s) = (d, e)$ and $\sigma(ss'\vec{s}) = \sigma_{s'}(s'\vec{s})$ for each $s' \in S$ and $\vec{s} \in S^*$
  where $\sigma_{s'}$ is a scheduler satisfying $k'(s') \in (\Prb{\sigma_{s'}}(s', \checkmark), \Erew{\sigma_{s'}}(s', \checkmark))^\downarrow$.
  The scheduler $\sigma$ provides
  \begin{align*}
    p(d, e, k') &= d(\checkmark) + \sum_{s'\in S} d(s')\cdot \pi_1  \big(k'(s')\big) \\
    &\leq d(\checkmark) + \sum_{s'\in S} d(s')\cdot \Prb{\sigma_{s'}}(s', \checkmark) \\
    &\qquad \text{by monotone convergence theorem} \\
    &= \sup_{n \in \mathbb{N}}\left(d(\checkmark) + \sum_{s'\in S} d(s')\cdot \Prb{\sigma_{s'}}^n(s', \checkmark)\right) \\
    &= \Prb{\sigma}(s, \checkmark),
  \end{align*}
  and
  \begin{align*}
    &r(d, e, k') \\
    =&\, e \cdot d(\checkmark) + \sum_{s'\in S} d(s')\cdot \left(\pi_2  \big(k'(s')\big) + e \cdot \pi_1  \big(k'(s')\big)\right) \\
    \leq&\, e \cdot d(\checkmark) + \sum_{s'\in S} d(s')\cdot \left(\Erew{\sigma_{s'}}(s', \checkmark) + e \cdot \Prb{\sigma_{s'}}(s', \checkmark)\right) \\
    \qquad& \text{by monotone convergence theorem} \\
    =& \sup_{n \in \mathbb{N}}\Big(e \cdot d(\checkmark) + \sum_{s'\in S} d(s')\cdot (\Erew{\sigma_{s'}}^n(s', \checkmark) \\
    &+ e \cdot \Prb{\sigma_{s'}}^n(s', \checkmark))\Big) \\
    =& \Erew{\sigma}(s, \checkmark).
  \end{align*}

  2) We prove it by induction on $n$.
  It's clear
  for the base case.
  For the step case, we assume that
  $\Phi^n(\Delta_{(0, 0)^\downarrow})(s) = \bigcup_{\text{scheduler }\sigma}\big(\Prb{\sigma}^{n}(s, \checkmark), \Erew{\sigma}^{n}(s, \checkmark)\big)^\downarrow$.
  \begin{align*}
    &\Phi^{n+1}(\Delta_{(0, 0)^\downarrow})(s)  \\
    &= \bigcup_{\substack{(d, e) \in c(s),
    k'\colon S \to \interval{1} \times \interval{\infty} \text{ s.t.~} \\\forall s'.~k'(s') \in \bigcup_{\sigma}(\Prb{\sigma}^n(s', \checkmark), \Erew{\sigma}^n(s', \checkmark))^\downarrow}}
    \big(p(d, e, k'), r(d, e, k')\big)^{\downarrow} \\
    &= \bigcup_{\substack{(d, e) \in c(s),
    \text{scheduler }\sigma}}
    \big(p(d, e, k_\sigma'), r(d, e, k'_\sigma)\big)^\downarrow \\
    &\qquad \text{where }k'_\sigma(s') = (\Prb{\sigma}^n(s', \checkmark), \Erew{\sigma}^n(s', \checkmark)) \\
    &= \bigcup_{\text{scheduler }\sigma}\big(\Prb{\sigma}^{n+1}(s, \checkmark), \Erew{\sigma}^{n+1}(s, \checkmark)\big)^\downarrow
  \end{align*}
  The second equation above can be proved in a similar manner as Lem.~\ref{lem:f_init}.\ref{item:f_alg}.

  3) Since $k \in \fixunit{[S, (\interval{1} \times \interval{\infty})^\downarrow]}$,
    $\Delta_{(0, 0)^\downarrow} \sqsubseteq k$ holds.
    Thus for each $n \in \mathbb{N}$,
    $\Phi^n(\Delta_{(0, 0)^\downarrow}) \sqsubseteq \Phi^n(k) \sqsubseteq k$
    by monotonicity of $\Phi$.

  4) Assume that $f \in \fixunit{[S, (\interval{1} \times \interval{\infty})^\downarrow]}$.
  Then it follows that $RL(\bigsqcup_{n \in \mathbb{N}}\Phi^n(\Delta_{(0, 0)^\downarrow})) = f$
  since both map $s$ to $(1, \bigsqcup_{\text{scheduler }\sigma}\Erew{\sigma}(s, \checkmark))^\downarrow$
  by Lem.~\ref{lem:f_init}.\ref{item:phi_delta}.
  Therefore,
for each $(k, \Phi k \sqsubseteq k) \in \AlgRst{\Phi}{\fixunit{[S, (\interval{1} \times \interval{\infty})^\downarrow]}}$,
\begin{align*}
  f &= RL(\bigsqcup_{n \in \mathbb{N}}\Phi^n(\Delta_{(0, 0)^\downarrow}))  \\
  &\sqsubseteq RLk &\text{by Lem.~\ref{lem:f_init}.\ref{item:cup_k}} \\
  &= k &\text{by }k \in \fixunit{[S, (\interval{1} \times \interval{\infty})^\downarrow]}.
\end{align*}
\end{proof}

\begin{proposition} \label{prop:gra_corresp_f}
  Consider Example~\ref{eg:achievable_mdp}.
Under the \gra of $(f, \Phi f \sqsubseteq f)$,
$(f, \Phi f \sqsubseteq f)$ corresponds to $\iota_{L \Phi R}$ via the equivalence~\eqref{eq:alg_equiv}.
\qed
\end{proposition}

\section{Definition of liftings for predicate transformers} \label{ap:fib_lift}
In \S\ref{subsec:corresp_lifting}, 
we demonstrated an instance of the initial algebra correspondence 
in a fibrational setting.
While our focus was on a single example,
the same approach is applicable to
the other examples in \S\ref{sec:init_alg_corresp} as well.
Here 
we introduce definitions of liftings $\dot{F}$ of $F$ along $\fibop{\laxdom{\Omega_1}}$,
which is necessary for extending these examples to the fibrational setting.

\begin{definition} 
  Consider Example~\ref{eg:total_partial_correct}.
  Instead of the coalgebra $c\colon S \to (S + \{\checkmark\})^\Sigma$,
  here we use the coalgebra $c'\colon S \times \Sigma^\omega \to S \times \Sigma^\omega + \mathbf{2}$ defined by $c'(s, \sigma\vec{\sigma}) \coloneqq (c(s)(\sigma), \vec{\sigma})$ if $c(s)(\sigma) \in S$, and $A(s)$ otherwise.

  Let $\Omega_1 \coloneqq \rs$, $\Omega_2 \coloneqq \mathbf{2}$,
  and
  $F \coloneqq ((-) + \mathbf{2})$.
  Define $\dot{F}\colon \fibop{(\Set/\Omega_1)} \to \fibop{(\Set/\Omega_1)}$ by
  \begin{displaymath}
   \dot{F}\big((X, k)\big) \coloneqq (FX, k') \text{ and } \dot{F}(h) \coloneqq Fh,
  \end{displaymath}
    where $k'\colon FX \to \Omega_1$ sends 
    $x \in X$ to $(\pi_1k(x), \neg \pi_1 k(x) \lor \pi_2 k(x))$ and
    each $b \in \mathbf{2}$ to $(\top, b)$.
\end{definition}

\begin{definition} 
  Consider Example~\ref{eg:cost-boundedReach}.
  Let $\Omega_1 \coloneqq M_\bot$, $\Omega_2 \coloneqq \mathbf{2}$,
  and
  $F \coloneqq \mathcal{P}(-) \times \mathbb{N} + \{\checkmark\}$.
  Define $\dot{F}\colon \fibop{(\Set/\Omega_1)} \to \fibop{(\Set/\Omega_1)}$ by
  \begin{displaymath}
   \dot{F}\big((X, k)\big) \coloneqq (FX, k') \text{ and } \dot{F}(h) \coloneqq Fh,
  \end{displaymath}
    where $k'\colon FX \to \Omega_1$ sends 
    $\checkmark$ to 0 and
    each $(I, n) \in \mathcal{P}X \times \mathbb{N}$ to
    \begin{align*}
   &\bigsqcup_{m \in k(I) \cap \mathbb{N}}\max(M, m+n).
    \end{align*}
\end{definition}

\begin{definition} 
  Consider Example~\ref{eg:achievable_mdp}.
  Let $\Omega_1 \coloneqq (\interval{1} \times \interval{\infty})^\downarrow$, $\Omega_2 \coloneqq \interval{\infty}$,
  and
  $F \coloneqq \Pfin(\mathcal{D}(- + \{\checkmark\}) \times \mathbb{N})$.
  Define $\dot{F}\colon \fibop{(\Set/\Omega_1)} \to \fibop{(\Set/\Omega_1)}$ by
  \begin{displaymath}
   \dot{F}\big((X, k)\big) \coloneqq (FX, k') \text{ and } \dot{F}(h) \coloneqq Fh,
  \end{displaymath}
    where $k'\colon FX \to \Omega_1$ sends each $I \in FX$ to
    \begin{align*}
    &\bigcup_{\substack{(d, n) \in I,
    j\colon X \to \interval{1} \times \interval{\infty} \\ \text{ s.t.~}\forall x \in X.~j(x) \in k(x)}}
    &\Big(\big(p(d, n, j), r(d, n, j)\big)\Big)^{\downarrow},
    \end{align*}
    and
    $p(d, n, j) \in \interval{1}$ and $r(d, n, j) \in \interval{\infty}$ are defined by
  \begin{align*}
    p(d, n, j) &\coloneqq d(\checkmark) + \sum_{x\in X} d(x)\cdot \pi_1  \big(j(x)\big), \\
    r(d, n, j) &\coloneqq
    n \cdot d(\checkmark)  \\
    &\phantom{\coloneqq}+ \sum_{x\in X} d(x)\cdot (\pi_2  \big(j(x)\big) + n \cdot \pi_1  \big(j(x)\big)).
  \end{align*}
\end{definition}

\section{Omitted proofs for \S\ref{sec:init_alg_corresp_L}} \label{ap:omitted_L}
\begin{proof}[Proof of Prop.~\ref{prop:init_alg_isom}]
  We first aim to show that
  $W_\rho$ is natural.
  We will show that for each ordinal $i, j$ with $j < i$, $(W_\rho)_i \circ W_{\Xi'}(j, i) = W_\Xi(j, i) \circ (W_\rho)_j$.
  We proceed by induction on $i$.
  \begin{itemize}
    \item Case $i=0$: There is no such $j$.
    \item Case $i$ is a limit ordinal: Clear by definition.
    \item Case $i$ is a successor ordinal:
    We prove it by induction on $j \ (< i)$.
    If $j=0$, it's clear.
    If $j$ is a successor ordinal,
    then
    \begin{align*}
      &(W_\rho)_i \circ W_{\Xi'}(j, i) \\
      &= \rho_{\Xi^{i-1} 0}\circ \Xi'((W_\rho)_{i-1}) \circ W_{\Xi'}(j, i) \\
      &= \rho_{\Xi^{i-1} 0}\circ \Xi'((W_\rho)_{i-1} \circ W_{\Xi'}(j-1, i-1)) \\
      &\qquad \text{by induction hypothesis} \\
      &= \rho_{\Xi^{i-1} 0}\circ \Xi'(W_\Xi(j-1, i-1) \circ (W_\rho)_{i-1}) \\
      &= W_\Xi(j, i) \circ \rho_{\Xi^{j-1} 0}\circ \Xi'((W_\rho)_{i-1}) \\
      &= W_\Xi(j, i) \circ (W_\rho)_j.
    \end{align*}
    If $j$ is a limit ordinal,
    it follows that $(W_\rho)_i \circ W_{}(j, i) = W_\Xi(j, i) \circ (W_\rho)_j$ by universality of the colimit $(\Xi')^j 0$: For each ordinal $j'$ with $j' < j$,
    \begin{align*}
      &(W_\rho)_i \circ W_{\Xi'}(j, i) \circ W_{\Xi'}(j', j) \\
      &= (W_\rho)_i \circ W_{\Xi'}(j', i) \\
      &\qquad \text{by induction hypothesis} \\
      &= W_\Xi(j', i) \circ (W_\rho)_{j'} \\
      &= W_\Xi(j, i) \circ W_\Xi(j', j) \circ (W_\rho)_{j'} \\
      &\qquad \text{by induction hypothesis} \\
      &= W_\Xi(j, i) \circ (W_\rho)_j \circ W_{\Xi'}(j', j).
    \end{align*}
  \end{itemize}
  Thus,
  $W_\rho\colon W_{\Xi'} \Rightarrow W_\Xi$ is a natural transformation.

  If $\rho_{\Xi^i 0}$ is an isomorphism for each ordinal $i$,
  then $W_\rho$ is an isomorphism by definition.
  Conversely,
  if $W_\rho$ is an isomorphism,
  then for each ordinal $i$,
  $\rho_{\Xi^i 0} = (W_\rho)_{i+1} \circ \big(\Xi'((W_\rho)_i)\big)^{-1}$
  is an isomorphism.
\end{proof}

\begin{proof}[Proof of Lem.~\ref{lem:init_alg_corresp_L}]
  Let $\lambda$ be the ordinal in which the initial chain of $\Psi$ converges.
  By Prop.~\ref{prop:init_alg_isom}, $W_\rho$ is an isomorphism.
  Since $(W_\rho)_\lambda\colon (L \Phi R)^\lambda 0 \cong \Psi^\lambda 0$ and
  the initial chain of $\Psi$ converges in $\lambda$ steps,
  the initial chain of $L \Phi R$ also converges in $\lambda$ steps.
  In particular,
  both $W_{L \Phi R}(\lambda, \lambda+1)^{-1}$
  and
  $W_{\Psi}(\lambda, \lambda+1)^{-1}$ are initial algebras for $L \Phi R$ and $\Psi$, respectively.
  Moreover, $W_{L \Phi R}(\lambda, \lambda+1)^{-1}$ is isomorphic to $\AlgFn{\id}{\rho}(W_\Psi(\lambda, \lambda+1)^{-1})$
  by the following equations.
  \begin{align*}
    &(W_\rho)_\lambda \circ W_{L \Phi R}(\lambda, \lambda+1)^{-1} \\
    &= W_\Psi(\lambda, \lambda+1)^{-1} \circ (W_\rho)_{\lambda+1} \\
    &= W_\Psi(\lambda, \lambda+1)^{-1} \circ \rho_{\Psi^\lambda 0} \circ L\Phi R((W_\rho)_\lambda) \\
    &= \AlgFn{\id}{\rho}(W_\Psi(\lambda, \lambda+1)^{-1}) \circ L\Phi R((W_\rho)_\lambda).
  \end{align*}
\end{proof}

\begin{proof}[Proof of Prop.~\ref{prop:init_alg_corresp_L}]
  Assume that conditions \eqref{item:rho_bar}, \eqref{item:converge}, \eqref{item:init_chain_fix} are satisfied.
  For each $i$,
  since
  $\rho_{\Psi^i 0} = \epsilon_{\Psi^{i+1} 0} \circ L\overline{\rho}_{\Psi^i 0}$
  (resp.~$\rho_{\Psi^i 0} = \Psi \epsilon_{\Psi^i 0} \circ \rho_{LR\Psi^i 0} \circ L\Phi R \epsilon_{\Psi^i 0}^{-1} =  \Psi \epsilon_{\Psi^i 0} \circ \rho'_{R \Psi^i 0}$ where $\rho' \coloneqq \rho_L \circ L \Phi \eta$),
  $\rho_{\Psi^i 0}$
  is an isomorphism.
  Hence Lem.~\ref{lem:init_alg_corresp_L} shows that
  $\AlgFn{\id}{\rho}(\iota_\Psi)$ is initial in $\Alg{L\Phi R}$.
  The conditions \eqref{item:converge} and \eqref{item:init_chain_fix} ensure $\mu \Psi \in \fixunit{\mathbb{D}}$.
  By Lem.~\ref{lem:gra_psi}, it follows that the initial algebra correspondence holds.

  Next, suppose $\mu \Psi \in \fixcounit{\mathbb{D}}$ (the \gra of $\iota_\Psi$) and $\fixcounit{\mathbb{D}}\hookrightarrow \mathbb{D}$ is downward closed.
  Because $\mu \Psi \in \fixcounit{\mathbb{D}}$,
  $\Psi^\lambda 0 \in \fixcounit{\mathbb{D}}$ where $\lambda$ is the ordinal in which the initial chain of $\Psi$ converges.
  By downward closure, the elements of the initial chain all lie in $\fixcounit{\mathbb{D}}$.
  Hence, \eqref{item:init_chain_fix} is satisfied.
\end{proof}

\begin{proposition} \label{ap:prop:unambiguous}
  In the setting of Example~\ref{eg:unambiguousFA},
  all three conditions of Prop.~\ref{prop:init_alg_corresp_L} are satisfied under the reachability condition of $\iota_\Psi$.
\end{proposition}
\begin{proof}
    \begin{itemize}
      \item Condition (1): It is clear because $\Phi R = R \Psi$.
      \item Condition (2):
        Since $\Psi$ is an $\omega$-continuous function on the complete lattice $[S \times A^*, \Ninf]$
        by the monotone convergence theorem,
        the initial chain of $\Psi$ converges in $\omega$ steps.
      \item Condition (3): Since $[S \times A^*, \Ninf]$ is downward closed, the \gra of $\iota_\Psi$ implies this condition by Prop.~\ref{prop:init_alg_corresp_L}.
    \end{itemize}
\end{proof}

\section{Another example about unambiguous automata} \label{ap:unambiguous_mc}
In the following example, we often see $\bot, \top$ as $0, 1$, respectively.
\begin{example}
  We consider an NFA under a probabilistic semantics induced by an MC.
  Specifically, let $\langle \delta, \Acc\rangle\colon S \to \mathcal{P}(S)^A \times \mathbf{2}$
  and $c\colon A \to \mathcal{D}(A + \{\checkmark\})$
  be coalgebras
  representing an NFA (with finite sets $S$ and $A$) and an MC, respectively.
  For a word $\vec{a} \in A^*$,
  we focus on two different computation for the probability of the path $\vec{a}\checkmark$ in the MC $c$.
  We show that,  under a \gra reflecting the unambiguity of NFAs, these two computations coincide.
  Such computations (for unambiguous NFAs) have been explored in the context of probabilistic verification~\cite{BaierK00023}.

  We define an endofunctor $\Phi$ on $[S \times A^*, \mathbf{2} \times [0, 1]]$ by
    \begin{align*}
      \Phi(k)(s, \vec{a})
      &\coloneqq
      \begin{cases}
        \big(\Acc(s), \Acc(s)\big) &\text{if }\vec{a} = \epsilon, \\
        \big(\bigvee_{s' \in \delta(s)(a)}\pi_1 k(s', \vec{a'}), p(s, a, \vec{a})) &\text{if  }\vec{a} = a\vec{a'},
      \end{cases} \\
    p(s, a, \vec{a}) &\coloneqq \mathrm{min}(\pi_1 k(s, \vec{a}), \bigvee_{s' \in \delta(s)(a)}c(a, b) \cdot \pi_2 k(s', \vec{a'})),
    \end{align*}
    where $b$ is the first element of $\vec{a'}\checkmark$.
    The domain $\mu \Phi$ characterizes both the recognized language of $c$ and the probability of accepted words.
    Specifically, $\pi_1 \mu \Phi(s, \vec{a})$
    indicates whether $\vec{a}$ is accepted from $s$,
    and $\pi_2 \mu \Phi(s, \vec{a})$ is the probability of $\vec{a}$
    if $\vec{a}$ is accepted from $s$ in the NFA, and is $0$ otherwise.

    Next, define an endofunctor $\Psi$ on $[S \times A^*, \Ninf \times \interval{\infty}]$ by
  \begin{align*}
    \Psi(k)(s, \vec{a})
    &\coloneqq \begin{cases}
   \big(\Acc(s), \Acc(s)\big) &\text{if } \vec{a} = \epsilon, \\
    \big(\Sigma_{s' \in (\delta(s))(a)}\pi_1k(s', \vec{a'}), p(s, a, \vec{a})\big) &\text{if } \vec{a} = a\vec{a'},
    \end{cases} \\
    p(s, a, \vec{a}) &\coloneqq \Sigma_{s' \in \delta(s)(a)}c(a, b) \cdot \pi_2k(s', \vec{a'}),
  \end{align*}
    where $b$ is the first element of $\vec{a'}\checkmark$.
  The domain $\mu \Psi$ represents both the number of accepting paths and the sum of the probabilities of those paths.

  By considering an adjunction defined as in Example~\ref{eg:unambiguousFA},
  Prop.~\ref{prop:init_alg_corresp_L} gives that
  the initial algebras of $\Phi$ and $\Psi$ coincide
  under the unambiguity condition of the NFA.
  The \gra $\mu \Psi \in \fixunit{\mathbb{D}}$ is the unambiguity of the NFA,
  and the resulting initial algebra correspondence shows
  the correspondence between the probability of accepted words and the sum of the probabilities of accepting paths.
\end{example}

\section{Dual version of Prop.~\ref{prop:fix_total}} \label{ap:dual_fix_total}

\begin{corollary} \label{cor:fix_coalg}
  Let $p\colon \mathbb{E} \to \mathbb{C}$ and $q\colon \mathbb{F} \to \mathbb{D}$ be functors.
  Suppose
  $\dot{\Phi}\colon \mathbb{E} \to \mathbb{E}$ is a lifting of $\Phi\colon \mathbb{C} \to \mathbb{C}$ along $p$,
  and an adjunction
  $\adj {\dot{L}} {\dot{R}} {\dot{\eta}} {\dot{\epsilon}} \colon {\mathbb{F}} \to {\mathbb{E}}$ is a lifting of $\adj L R \eta \epsilon \colon {\mathbb{D}} \to {\mathbb{C}}$ along $q, p$.
  Applying Prop.~\ref{prop:fix_total} to the opposites of them,
  we obtain
  the adjoint equivalence
  \begin{displaymath}
    \adj {\CoalgFn{\dot{L}}{\dot{\alpha}}} {\CoalgFn{\dot{R}}{\dot{\beta}}} {\dot{\eta}} {\dot{\epsilon}}\colon \CoalgRst{\dot{R}\dot{\Phi}\dot{L}}{\fixunit{\mathbb{F}}} \to \CoalgRst{\dot{\Phi}}{\fixcounit{\mathbb{E}}},
  \end{displaymath}
  which is a lifting of the adjoint equivalence
  \begin{displaymath}
    \adj {\CoalgFn{L}{\alpha}} {\CoalgFn{R}{\beta}} \eta \epsilon \colon\CoalgRst{R\Phi L}{\fixunit{\mathbb{D}}} \to  \CoalgRst{\Phi}{\fixcounit{\mathbb{C}}}
  \end{displaymath}
  along $\Coalg{q}, \Coalg{p}$, where
  $\dot{\alpha} = \dot{\epsilon} \dot{\Phi}\dot{L}$,
  $\dot{\beta} = \dot{R}\dot{\Phi}\dot{\epsilon}^{-1}$,
  $\alpha = \epsilon \Phi L$, and
  $\beta = R\Phi \epsilon^{-1}$.
  \begin{displaymath}
    \xymatrix@R=3em{
      \CoalgRst{\dot{\Phi}}{\fixcounit{\mathbb{E}}} \ar[d]^{\Coalg{p}}  \ar@/^1em/[rr]^{\CoalgFn{\dot{R}}{\dot{R}\dot{\Phi}\dot{\epsilon}^{-1}}} &\simeq &\CoalgRst{\dot{R}\dot{\Phi}\dot{L}}{\fixunit{\mathbb{F}}} \ar@/^1em/[ll]^{\CoalgFn{\dot{L}}{\dot{\epsilon} \dot{\Phi}\dot{L}}} \ar[d]^{\Coalg{q}} \\
      \CoalgRst{\Phi}{\fixcounit{\mathbb{C}}}   \ar@/^1em/[rr]^{\CoalgFn{R}{R\Phi \epsilon^{-1}}} &\simeq &\CoalgRst{R\Phi L}{\fixunit{\mathbb{D}}} \ar@/^1em/[ll]^{\CoalgFn{L}{\epsilon \Phi L}}
    }
  \end{displaymath}
\end{corollary}

%% file: main.bbl
\begin{thebibliography}{10}
\providecommand{\url}[1]{#1}
\csname url@samestyle\endcsname
\providecommand{\newblock}{\relax}
\providecommand{\bibinfo}[2]{#2}
\providecommand{\BIBentrySTDinterwordspacing}{\spaceskip=0pt\relax}
\providecommand{\BIBentryALTinterwordstretchfactor}{4}
\providecommand{\BIBentryALTinterwordspacing}{\spaceskip=\fontdimen2\font plus
\BIBentryALTinterwordstretchfactor\fontdimen3\font minus \fontdimen4\font\relax}
\providecommand{\BIBforeignlanguage}[2]{{%
\expandafter\ifx\csname l@#1\endcsname\relax
\typeout{** WARNING: IEEEtran.bst: No hyphenation pattern has been}%
\typeout{** loaded for the language `#1'. Using the pattern for}%
\typeout{** the default language instead.}%
\else
\language=\csname l@#1\endcsname
\fi
#2}}
\providecommand{\BIBdecl}{\relax}
\BIBdecl

\bibitem{Puterman94}
M.~L. Puterman, \emph{Markov Decision Processes: Discrete Stochastic Dynamic Programming}, ser. Wiley Series in Probability and Statistics.\hskip 1em plus 0.5em minus 0.4em\relax Wiley, 1994.

\bibitem{DBLP:books/daglib/0020348}
C.~Baier and J.~Katoen, \emph{Principles of model checking}.\hskip 1em plus 0.5em minus 0.4em\relax {MIT} Press, 2008.

\bibitem{DBLP:conf/qest/GretzKM12}
F.~Gretz, J.~Katoen, and A.~McIver, ``Operational versus weakest precondition semantics for the probabilistic guarded command language,'' in \emph{{QEST}}.\hskip 1em plus 0.5em minus 0.4em\relax {IEEE} Computer Society, 2012, pp. 168--177.

\bibitem{OlmedoGJKKM18}
F.~Olmedo, F.~Gretz, N.~Jansen, B.~L. Kaminski, J.~Katoen, and A.~McIver, ``Conditioning in probabilistic programming,'' \emph{{ACM} Trans. Program. Lang. Syst.}, vol.~40, no.~1, pp. 4:1--4:50, 2018.

\bibitem{KaminskiKMO18}
B.~L. Kaminski, J.~Katoen, C.~Matheja, and F.~Olmedo, ``Weakest precondition reasoning for expected runtimes of randomized algorithms,'' \emph{J. {ACM}}, vol.~65, no.~5, pp. 30:1--30:68, 2018.

\bibitem{ChenFKPS13}
T.~Chen, V.~Forejt, M.~Z. Kwiatkowska, D.~Parker, and A.~Simaitis, ``Automatic verification of competitive stochastic systems,'' \emph{Formal Methods Syst. Des.}, vol.~43, no.~1, pp. 61--92, 2013.

\bibitem{DBLP:books/cu/J2016}
B.~Jacobs, \emph{Introduction to Coalgebra: Towards Mathematics of States and Observation}, ser. Cambridge Tracts in Theoretical Computer Science.\hskip 1em plus 0.5em minus 0.4em\relax Cambridge University Press, 2016, vol.~59.

\bibitem{cousot21}
P.~Cousot, \emph{Principles of Abstract Interpretation}.\hskip 1em plus 0.5em minus 0.4em\relax MIT Press, 2021.

\bibitem{DBLP:journals/iandc/HermidaJ98}
C.~Hermida and B.~Jacobs, ``Structural induction and coinduction in a fibrational setting,'' \emph{Inf. Comput.}, vol. 145, no.~2, pp. 107--152, 1998.

\bibitem{DBLP:conf/popl/CousotC79}
P.~Cousot and R.~Cousot, ``Systematic design of program analysis frameworks,'' in \emph{{POPL}}.\hskip 1em plus 0.5em minus 0.4em\relax {ACM} Press, 1979, pp. 269--282.

\bibitem{DBLP:journals/jacm/GiacobazziRS00}
R.~Giacobazzi, F.~Ranzato, and F.~Scozzari, ``Making abstract interpretations complete,'' \emph{J. {ACM}}, vol.~47, no.~2, pp. 361--416, 2000.

\bibitem{BruniGGR21}
R.~Bruni, R.~Giacobazzi, R.~Gori, and F.~Ranzato, ``A logic for locally complete abstract interpretations,'' in \emph{{LICS}}.\hskip 1em plus 0.5em minus 0.4em\relax {IEEE}, 2021, pp. 1--13.

\bibitem{BruniGGR22}
------, ``Abstract interpretation repair,'' in \emph{{PLDI}}.\hskip 1em plus 0.5em minus 0.4em\relax {ACM}, 2022, pp. 426--441.

\bibitem{DBLP:conf/calco/TurkenburgBKR23}
R.~Turkenburg, H.~Beohar, C.~Kupke, and J.~Rot, ``Forward and backward steps in a fibration,'' in \emph{{CALCO}}, ser. LIPIcs, vol. 270.\hskip 1em plus 0.5em minus 0.4em\relax Schloss Dagstuhl - Leibniz-Zentrum f{\"{u}}r Informatik, 2023, pp. 6:1--6:18.

\bibitem{DBLP:journals/logcom/SprungerKDH21}
\BIBentryALTinterwordspacing
D.~Sprunger, S.~Katsumata, J.~Dubut, and I.~Hasuo, ``Fibrational bisimulations and quantitative reasoning: Extended version,'' \emph{J. Log. Comput.}, vol.~31, no.~6, pp. 1526--1559, 2021. [Online]. Available: \url{https://doi.org/10.1093/logcom/exab051}
\BIBentrySTDinterwordspacing

\bibitem{DBLP:conf/fossacs/WissmannDKH19}
T.~Wi{\ss}mann, J.~Dubut, S.~Katsumata, and I.~Hasuo, ``Path category for free - open morphisms from coalgebras with non-deterministic branching,'' in \emph{FoSSaCS}, ser. Lecture Notes in Computer Science, vol. 11425.\hskip 1em plus 0.5em minus 0.4em\relax Springer, 2019, pp. 523--540.

\bibitem{Wissmann2019}
\BIBentryALTinterwordspacing
T.~Wißmann, S.~Milius, S.-y. Katsumata, and J.~Dubut, ``\BIBforeignlanguage{eng}{A coalgebraic view on reachability},'' \emph{\BIBforeignlanguage{eng}{Commentationes Mathematicae Universitatis Carolinae}}, vol.~60, no.~4, pp. 605--638, 2019. [Online]. Available: \url{http://eudml.org/doc/295060}
\BIBentrySTDinterwordspacing

\bibitem{Batz2025}
K.~Batz, B.~L. Kaminski, C.~Matheja, and T.~Winkler, ``{J-P:} {MDP.} {FP.} {PP} - characterizing total expected rewards in {M}arkov decision processes as least fixed points with an application to operational semantics of probabilistic programs,'' in \emph{Principles of Verification {(1)}}, ser. Lecture Notes in Computer Science, vol. 15260.\hskip 1em plus 0.5em minus 0.4em\relax Springer, 2025, pp. 255--302.

\bibitem{Baier0KW17}
C.~Baier, J.~Klein, S.~Kl{\"{u}}ppelholz, and S.~Wunderlich, ``Maximizing the conditional expected reward for reaching the goal,'' in \emph{{TACAS} {(2)}}, ser. Lecture Notes in Computer Science, vol. 10206, 2017, pp. 269--285.

\bibitem{PiribauerB19}
J.~Piribauer and C.~Baier, ``Partial and conditional expectations in markov decision processes with integer weights,'' in \emph{FoSSaCS}, ser. Lecture Notes in Computer Science, vol. 11425.\hskip 1em plus 0.5em minus 0.4em\relax Springer, 2019, pp. 436--452.

\bibitem{DBLP:journals/logcom/RotJL21}
J.~Rot, B.~Jacobs, and P.~B. Levy, ``Steps and traces,'' \emph{J. Log. Comput.}, vol.~31, no.~6, pp. 1482--1525, 2021.

\bibitem{DBLP:journals/jcss/Jacobs0S15}
B.~Jacobs, A.~Silva, and A.~Sokolova, ``Trace semantics via determinization,'' \emph{J. Comput. Syst. Sci.}, vol.~81, no.~5, pp. 859--879, 2015.

\bibitem{DBLP:conf/fossacs/BonsangueK05}
M.~M. Bonsangue and A.~Kurz, ``Duality for logics of transition systems,'' in \emph{FoSSaCS}, ser. Lecture Notes in Computer Science, vol. 3441.\hskip 1em plus 0.5em minus 0.4em\relax Springer, 2005, pp. 455--469.

\bibitem{DBLP:conf/amast/PavlovicMW06}
D.~Pavlovic, M.~W. Mislove, and J.~Worrell, ``Testing semantics: Connecting processes and process logics,'' in \emph{{AMAST}}, ser. Lecture Notes in Computer Science, vol. 4019.\hskip 1em plus 0.5em minus 0.4em\relax Springer, 2006, pp. 308--322.

\bibitem{DBLP:journals/entcs/Klin07}
B.~Klin, ``Coalgebraic modal logic beyond sets,'' in \emph{{MFPS}}, ser. Electronic Notes in Theoretical Computer Science, vol. 173.\hskip 1em plus 0.5em minus 0.4em\relax Elsevier, 2007, pp. 177--201.

\bibitem{porst1991concrete}
H.-E. Porst and W.~Tholen, ``Concrete dualities,'' \emph{Category theory at work}, vol.~18, pp. 111--136, 1991.

\bibitem{DBLP:journals/dm/Baranga91}
A.~Baranga, ``The contraction principle as a particular case of {K}leene's fixed point theorem,'' \emph{Discret. Math.}, vol.~98, no.~1, pp. 75--79, 1991.

\bibitem{pjm/1103044538}
A.~Tarski, ``{A lattice-theoretical fixpoint theorem and its applications.}'' \emph{Pacific Journal of Mathematics}, vol.~5, no.~2, pp. 285 -- 309, 1955.

\bibitem{DBLP:books/daglib/0023251}
B.~Jacobs, \emph{Categorical Logic and Type Theory}, ser. Studies in logic and the foundations of mathematics.\hskip 1em plus 0.5em minus 0.4em\relax North-Holland, 2001, vol. 141.

\bibitem{DBLP:conf/popl/CousotC77}
P.~Cousot and R.~Cousot, ``Abstract interpretation: {A} unified lattice model for static analysis of programs by construction or approximation of fixpoints,'' in \emph{{POPL}}.\hskip 1em plus 0.5em minus 0.4em\relax {ACM}, 1977, pp. 238--252.

\bibitem{stella1998optimization}
S.~X.~Yu, Y.~Lin, and P.~Yan, ``Optimization models for the first arrival target distribution function in discrete time,'' \emph{Journal of mathematical analysis and applications}, vol. 225, no.~1, pp. 193--223, 1998.

\bibitem{ChristmanC13}
A.~Christman and J.~Cassamano, ``Maximizing the probability of arriving on time,'' in \emph{{ASMTA}}, ser. Lecture Notes in Computer Science, vol. 7984.\hskip 1em plus 0.5em minus 0.4em\relax Springer, 2013, pp. 142--157.

\bibitem{HartmannsJKQ20}
A.~Hartmanns, S.~Junges, J.~Katoen, and T.~Quatmann, ``Multi-cost bounded tradeoff analysis in {MDP},'' \emph{J. Autom. Reason.}, vol.~64, no.~7, pp. 1483--1522, 2020.

\bibitem{ChenFKSW13}
T.~Chen, V.~Forejt, M.~Z. Kwiatkowska, A.~Simaitis, and C.~Wiltsche, ``On stochastic games with multiple objectives,'' in \emph{{MFCS}}, ser. Lecture Notes in Computer Science, vol. 8087.\hskip 1em plus 0.5em minus 0.4em\relax Springer, 2013, pp. 266--277.

\bibitem{AshokCKWW20}
P.~Ashok, K.~Chatterjee, J.~Kret{\'{\i}}nsk{\'{y}}, M.~Weininger, and T.~Winkler, ``Approximating values of generalized-reachability stochastic games,'' in \emph{{LICS}}.\hskip 1em plus 0.5em minus 0.4em\relax {ACM}, 2020, pp. 102--115.

\bibitem{Adamek1974}
\BIBentryALTinterwordspacing
J.~Adámek, ``\BIBforeignlanguage{eng}{Free algebras and automata realizations in the language of categories},'' \emph{\BIBforeignlanguage{eng}{Commentationes Mathematicae Universitatis Carolinae}}, vol. 015, no.~4, pp. 589--602, 1974. [Online]. Available: \url{http://eudml.org/doc/16649}
\BIBentrySTDinterwordspacing

\bibitem{Colcombet12}
T.~Colcombet, ``Forms of determinism for automata (invited talk),'' in \emph{{STACS}}, ser. LIPIcs, vol.~14.\hskip 1em plus 0.5em minus 0.4em\relax Schloss Dagstuhl - Leibniz-Zentrum f{\"{u}}r Informatik, 2012, pp. 1--23.

\bibitem{Colcombet15}
------, ``Unambiguity in automata theory,'' in \emph{{DCFS}}, ser. Lecture Notes in Computer Science, vol. 9118.\hskip 1em plus 0.5em minus 0.4em\relax Springer, 2015, pp. 3--18.

\bibitem{BaierK00023}
C.~Baier, S.~Kiefer, J.~Klein, D.~M{\"{u}}ller, and J.~Worrell, ``Markov chains and unambiguous automata,'' \emph{J. Comput. Syst. Sci.}, vol. 136, pp. 113--134, 2023.

\bibitem{StearnsH85}
R.~E. Stearns and H.~B.~H. III, ``On the equivalence and containment problems for unambiguous regular expressions, regular grammars and finite automata,'' \emph{{SIAM} J. Comput.}, vol.~14, no.~3, pp. 598--611, 1985.

\bibitem{Seidl90}
H.~Seidl, ``Deciding equivalence of finite tree automata,'' \emph{{SIAM} J. Comput.}, vol.~19, no.~3, pp. 424--437, 1990.

\bibitem{DBLP:conf/dcfs/Colcombet15}
T.~Colcombet, ``Unambiguity in automata theory,'' in \emph{{DCFS}}, ser. Lecture Notes in Computer Science, vol. 9118.\hskip 1em plus 0.5em minus 0.4em\relax Springer, 2015, pp. 3--18.

\bibitem{mac2013categories}
S.~Mac~Lane, \emph{Categories for the working mathematician}.\hskip 1em plus 0.5em minus 0.4em\relax Springer Science \& Business Media, 2013, vol.~5.

\bibitem{Loregian_2021}
\BIBentryALTinterwordspacing
F.~Loregian, \emph{(Co)end Calculus}.\hskip 1em plus 0.5em minus 0.4em\relax Cambridge University Press, Jun. 2021. [Online]. Available: \url{http://dx.doi.org/10.1017/9781108778657}
\BIBentrySTDinterwordspacing

\bibitem{DBLP:journals/mscs/HasuoKC18}
I.~Hasuo, T.~Kataoka, and K.~Cho, ``Coinductive predicates and final sequences in a fibration,'' \emph{Math. Struct. Comput. Sci.}, vol.~28, no.~4, pp. 562--611, 2018.

\bibitem{TuriP97}
D.~Turi and G.~D. Plotkin, ``Towards a mathematical operational semantics,'' in \emph{{LICS}}.\hskip 1em plus 0.5em minus 0.4em\relax {IEEE} Computer Society, 1997, pp. 280--291.

\bibitem{DBLP:journals/pacmpl/WatanabeJRH25}
K.~Watanabe, S.~Junges, J.~Rot, and I.~Hasuo, ``A unifying approach to product constructions for quantitative temporal inference,'' \emph{Proc. {ACM} Program. Lang.}, vol.~9, no. {OOPSLA1}, pp. 1575--1603, 2025.

\bibitem{DBLP:conf/concur/KoriHK21}
M.~Kori, I.~Hasuo, and S.~Katsumata, ``Fibrational initial algebra-final coalgebra coincidence over initial algebras: Turning verification witnesses upside down,'' in \emph{{CONCUR}}, ser. LIPIcs, vol. 203.\hskip 1em plus 0.5em minus 0.4em\relax Schloss Dagstuhl - Leibniz-Zentrum f{\"{u}}r Informatik, 2021, pp. 21:1--21:22.

\bibitem{ForejtKNP11}
V.~Forejt, M.~Z. Kwiatkowska, G.~Norman, and D.~Parker, ``Automated verification techniques for probabilistic systems,'' in \emph{{SFM}}, ser. Lecture Notes in Computer Science, vol. 6659.\hskip 1em plus 0.5em minus 0.4em\relax Springer, 2011, pp. 53--113.

\end{thebibliography}
